\documentclass[10pt,reqno]{amsart}
\usepackage{textcomp}
\usepackage{graphicx} 
\usepackage{float}
\usepackage{listings}
\usepackage{fancyvrb}
\usepackage{fvextra}
\usepackage{amssymb}

\usepackage[a4paper, left=30mm,right=30mm,top=30mm,bottom=30mm,marginpar=25mm]{geometry}
\usepackage[english]{babel}
\usepackage{setspace}
\usepackage{amsmath}
\usepackage{mathtools}
\usepackage{hyperref}
\usepackage[normalem]{ulem}
\usepackage{amsfonts}
\usepackage{amssymb}
\usepackage{fancyhdr}
\usepackage{color}

\newtheorem{theorem}{Theorem}
\newtheorem{lemma}[theorem]{Lemma}
\newtheorem{proposition}[theorem]{Proposition}
\newtheorem{remark}[theorem]{Remark}
\newtheorem{remarks}[theorem]{Remarks}
\newtheorem{corollary}[theorem]{Corollary}
\newcommand{\ds}{\displaystyle}
\newcommand{\Tn}[0]{{Ti\textsubscript{76}Nb\textsubscript{22}Al\textsubscript{2}}}

\newcommand{\R}{\mathbb R}
\newcommand{\bnu}{\boldsymbol{\nu}}

\newcommand{\mS}{\mathcal S}
\newcommand{\cof}{\hbox{cof}\,}

\newcommand{\bR}{\boldsymbol{R}}
\newcommand{\bP}{\boldsymbol{P}}
\newcommand{\bQ}{\boldsymbol{Q}}
\newcommand{\bA}{\boldsymbol{A}}
\newcommand{\bX}{\boldsymbol{X}}
\newcommand{\bB}{\boldsymbol{B}}
\newcommand{\bb}{\boldsymbol{b}}
\newcommand{\bc}{\boldsymbol{c}}
\newcommand{\bU}{\boldsymbol{U}}
\newcommand{\bF}{\boldsymbol{F}}
\newcommand{\bG}{\boldsymbol{G}}
\newcommand{\bD}{\boldsymbol{D}}

\newcommand{\be}{\boldsymbol{e}}
\newcommand{\bE}{\boldsymbol{E}}
\newcommand{\bm}{\boldsymbol{m}}
\newcommand{\bM}{\boldsymbol{M}}
\newcommand{\bcM}{\boldsymbol{\mathcal{M}}}

\newcommand{\dist}{{\rm dist}\,}

\newcommand{\bu}{\boldsymbol{u}}
\newcommand{\bp}{\boldsymbol{p}}
\newcommand{\bN}{\boldsymbol{N}}
\newcommand{\bK}{\boldsymbol{K}}
\newcommand{\bV}{\boldsymbol{V}}

\newcommand{\bw}{\boldsymbol{w}}
\newcommand{\bx}{\boldsymbol{x}}
\newcommand{\by}{\boldsymbol{y}}
\newcommand{\bz}{\boldsymbol{z}}
\newcommand{\bn}{\boldsymbol{n}}
\newcommand{\ba}{\boldsymbol{a}} 

\newcommand{\bC}{\boldsymbol{C}}
\newcommand{\bS}{\boldsymbol{S}}
\newcommand{\bI}{\boldsymbol{1}}
\newcommand{\bo}{\boldsymbol{0}}
\newcommand{\bal}{\boldsymbol{\alpha}}
\newcommand{\bet}{\boldsymbol{\beta}}
\newcommand{\tr}{\mathrm{tr}\,}
\newcommand{\di}{\mathrm{diag}\,}

\newcommand{\ep}{\varepsilon}
\newcommand{\diag}{{\rm diag}\,}
\newcommand{\diam}{{\rm diam}\,}
\newcommand{\supp}{{\rm supp}\,}
\newcommand{\rank}{{\rm rank}\,}
\newcommand{\om}{\Omega}
\newcommand{\bzero}{\mathbf 0}
\newcommand{\weakstar}{\stackrel{*}{\rightharpoonup}}
\newcommand{\av}{-\hspace{-.15in}\int}

\definecolor{listinggray}{gray}{0.9}
\definecolor{lbcolor}{rgb}{0.9,0.9,0.9}

\allowdisplaybreaks
\begin{document}
\title[Compatibility in Polycrystals]{Compatibility of Martensitic Microstructures in Polycrystals}

\author[J.M. Ball]{John M. Ball}
\address[John M. Ball]{\newline
Department of Mathematics, Heriot-Watt University\\ 
Edinburgh, Scotland, UK}
\email[]{\href{John.Ball@hw.ac.uk}{John.Ball@hw.ac.uk}}

\author[M. Galanopoulou]{Myrto Galanopoulou}
\address[Myrto Galanopoulou]{\newline Department of Mathematics, University of Sussex\\
Brighton, UK}
\email[]{\href{m.m.galanopoulou@sussex.ac.uk}{m.m.galanopoulou@sussex.ac.uk}, \href{myrtomaria.galanopoulou@alumni.kaust.edu.sa}{myrtomaria.galanopoulou@alumni.kaust.edu.sa}}

\begin{abstract}
The paper studies martensitic microstructures in polycrystals, focussing on their compatibility across grain boundaries. After a reduction to the case of a planar grain boundary, the case when the grain boundary separates two constant gradients of zero energy is considered. It is shown that for cubic-to tetragonal transformations such a configuration  can occur when the relative grain rotation is not in the cubic group. Then the case when the grain boundary separates two simple laminates of zero energy is considered, it being shown using a computer-assisted symbolic calculation that in the cubic-to-tetragonal  case  compatibility is only possible for a closed set of measure zero in the manifold of grain boundary normals and relative grain rotations, and that a similar slightly weaker result holds for cubic-to-orthorhombic transformations. The results suggest why higher-order laminates are often observed for such transformations.

The Taylor set of deformation gradients is defined and studied, this set having the property that any deformation whose gradient belongs to it corresponds to a zero-energy microstructure for the polycrystal independent of grain geometry and grain rotations. New upper bounds for the Taylor set are proved for cubic-to-tetragonal and cubic-to-orthorhombic transformations, generalizing those of Bhattacharya \& Kohn \cite{bhattkohn96,bhattkohn97} using the geometrically linearized theory. We give a simple proof of a related result of Peigney \cite{Peigney2013} characterizing the positive diagonal matrices in $K^{\rm qc}$ for cubic-to-tetragonal transformations.
\vspace{0.1cm}
		
\vspace{0.1cm}\\
\noindent\textsc{MSC2020 subclasses}: 74N15, 74B20
\end{abstract}
	
\maketitle

\centerline{\it Dedicated to Robin Knops with gratitude and admiration}
\section{Introduction}
In this paper we study microstructure arising from martensitic phase transformations in polycrystals, focussing on the compatibility of such microstructures across grain boundaries.  Martensitic phase transformations are diffusionless solid-solid phase transformations in which  the underlying crystal lattice of an alloy changes shape at a critical temperature $\theta_c$, with the high temperature austenite phase typically having greater symmetry than  the low temperature martensite phase. The different grains of the polycrystal  correspond to different orientations of the crystal lattice. We concentrate on the case when the austenite has cubic symmetry.

In a single crystal the microstructure resulting from a martensitic phase transformation results from the requirement of compatibility between the different symmetry related variants of the martensite, and from the nucleation process, during which the martensite must be compatible with the austenite. However we will not be concerned with the nucleation process, and will assume that the temperature is below $\theta_c$, so that no austenite is present. 

In a polycrystal the microstructure in each grain must also be compatible across its boundary with the microstructure in neighbouring grains. This can lead to deformation of the grain boundaries, resulting in correlation of microstructural features in different  grains, as  seen, for example, in the study of Arlt \cite{Arlt90} of ${\rm BaTIO_3}$ in which laminates propagate from one grain to another.

For a single crystal occupying the bounded domain $\om\subset\R^3$ in the reference configuration the nonlinear elasticity model (see \cite{j32,j40,03bhattacharya}) is based on the free energy functional
$$I(\by,\theta)=\int_\om \psi(D\by(\bx),\theta)\,d\bx,$$
where $\by:\om\to \R^3$ is the deformation. For a fixed temperature $\theta<\theta_c$ the set $K=K(\theta)$ of deformation gradients of minimum energy (without loss of generality taken to be zero) has the form
$$K=\bigcup_{i=1}^NSO(3)\bU_i,$$
where the $\bU_i=\bU_i^T>\bzero$ are the $N$ variants of martensite. For a polycrystal with $\om={\rm int}\,\bigcup_{j=1}^M\bar\om_j,$ having disjoint grains $\om_j$,  the corresponding free energy is given by
$$I(\by)=\int_\om\psi(D\by(\bx)\bR(\bx))\,d\bx,$$
where the grain rotation $\bR(\bx)=\bR_j\in SO(3)$ for $\bx\in\om_j$.

We first study the local problem of compatibility at zero energy across a single grain boundary separating two grains, the second having cubic axes rotated with respect to the first. Locally the grain boundary is described by its normal $\bm$ and the relative rotation $\bR\in SO(3)$ of the cubic axes in the second grain with respect to the first. By blowing up around a point on the grain boundary the problem of compatibility can be reduced under a regularity assumption to that of a planar grain boundary having normal $\bm$. The simplest case is that of zero-order laminates, that is when  the planar grain boundary separates constant gradients $\bA,\bB \in K$. There are then two subcases (i) when $\bA=\bB$, (ii) when $\bA\neq\bB$. We show that case (i) can occur for relative rotations $\bR$ not in the cubic group $P^{24}$ when two of the eigenvalues of the $\bU_i$ are equal, as for cubic-to-tetragonal transformations, but that this cannot happen for cubic-to-orthorhombic transformations. For case (ii) we give a criterion for compatibility (Theorem \ref{Theorem1}) in terms of the axis-angle representation of an orthogonal matrix related to the relative rotation; in particular this implies  (Corollary \ref{Cor2}) that when the eigenvalues of the $\bU_i$ are distinct compatibility fails for a dense open set of relative rotations and any $\bm$.

We then investigate whether compatibility can be achieved at zero energy  with a simple laminate on each side of the grain boundary. We show (Theorem \ref{compatthm}) that for cubic-to-tetragonal transformations with any given deformation parameters $\eta_1>0, \eta_2>0, \eta_1\neq\eta_2$ the set of $(\bm,\bR)\in M:=S^2\times SO(3)$ for which compatibility is possible is closed and of measure zero. Essentially this is a consequence of the condition for compatibility being an equation in $\R^{3\times 3}$, equivalently nine scalar equations, which has to be solved for eight unknowns (two volume fractions, a rotation and an amplitude vector). The condition for this to be possible is the vanishing of a product $\mathcal{F}$ of Dixon resultants formed from triple scalar products of vectors. To show that $\mathcal F$ is not identically zero we use computer assisted symbolic computation, and the result then follows from a general property of real analytic functions. For cubic-to-orthorhombic transformations we prove the slightly weaker result that for a.e. set of deformation parameters $(\alpha,\beta,\gamma)$ the set of  $(\bm,\bR)\in M:=S^2\times SO(3)$ for which compatibility is possible is closed and of measure zero. We conjecture (see Section \ref{3.4}) that, for example, compatibility can be achieved for generic $(\bm,\bR)$ with a simple laminate on one side of a planar grain boundary and a double laminate on the other side.

In order to give some information on possible zero-energy microstructures for polycrystals we investigate the Taylor set
\[\mathcal{S}(K):=\bigcap\limits_{\bR\in SO(3)} K^{\rm qc}\bR.\]
A related set was defined in a geometrically linearized setting by Bhattacharya \& Kohn \cite{bhattkohn96,bhattkohn97} and called by them the Taylor estimate on the recoverable strain.  Here $K^{\rm qc}$ denotes the quasiconvexification of $K$, that is the set of macroscopic deformation gradients corresponding to zero-energy microstructures. Any deformation with $D\by(\bx)\in\mathcal{S}(K)$ a.e. in $\om$ corresponds to a zero-energy microstructure {\it independent of grain geometry and grain rotations}. For cubic-to-tetragonal transformations Dolzmann \& Kirchheim \cite{DoKir} showed that
 $\mathcal{S}(K)$ contains a relatively open neighbourhood of the dilatation $(\eta_1^2\eta_2)^{1/3}\bI$ in $\{\bA\in \R^{3\times 3}:\det \bA=\eta_1^2\eta_2\}$. As a result there are nontrivial zero energy polycrystalline microstructures in any polycrystal for cubic-to-tetragonal transformations, and more generally whenever $N>1$ (see Section \ref{defprops}).

Since when $N\geq 3$ there is no known characterization of  $K^{\rm qc}$, it is not in general obvious how to calculate $\mS(K)$, which (unlike $K^{\rm qc}$) is isotropic. This isotropy implies that $\mS(K)=SO(3)\Delta(K)SO(3)$, where $\Delta(K)$ is the set of positive diagonal matrices $\bD=\diag(v_1,v_2,v_3)$ in $\mS(K)$.  However for the case of  two energy wells (which cannot occur for martensitic phase transformations with cubic austenite)
$$K=SO(3)\bU_1\cup SO(3)\bU_2,$$
with $\bU_1=\diag(\eta_1,\eta_2,\eta_3)$, $\bU_2=\diag(\eta_2,\eta_1,\eta_3)$ and  $\eta_2>\eta_1>0, \eta_3>0$ we can use the characterization of $K^{\rm qc}$ in \cite{j40},  so that $\mS(K)$ can be calculated. We do this in Theorem \ref{2DSK} together with a related characterization (Theorem \ref{2DSK1}) for the case when the grain rotations have the common axis $\be_3$.

In  the cubic-to-tetragonal case we prove, using the minors relations and properties of $K^{\rm c}$ and $(\cof K)^{\rm c}$, new upper bounds for $\mS(K)$, in particular showing (Theorem \ref{better}) that  $\bD\in\Delta(K)$ implies $v_1v_2v_3=\eta_1^2\eta_2$ and
\begin{equation}
\label{Dbounds}
\frac{\sqrt 3\,\eta_1\eta_2}{\sqrt{\eta_1^2+2\eta_2^2}}\leq v_i\leq\frac{\sqrt{\eta_2^2+2\eta_1^2}}{\sqrt 3},
\end{equation}
with a similar result (Theorem \ref{betteror}) for cubic-to-orthorhombic transformations.
 An interesting  result of Peigney \cite{Peigney2013} characterizing the positive diagonal matrices in $K^{\rm qc}$ for cubic-to-tetragonal transformations can also be used to give a bound on $\mS(K)$, but this turns out to be weaker than \eqref{Dbounds}. Nevertheless we state Peigney’s result for its independent interest as Theorem \ref{peigney} and give a simple proof.

While our results on compatibility of laminates at grain boundaries and on the Taylor set are developed independently, they of course both form part of the general question of understanding the complexity of zero-energy microstructures in polycrystals, as discussed in Section \ref{disc}.

  We organise the paper as follows:  In Section \ref{martensite} we review the nonlinear elasticity model of martensitic phase transformations as applied to single crystals and polycrystals. In Section \ref{3} we discuss compatibility of zero-energy microstructures across a grain boundary, first giving conditions under which the problem can be reduced to that for a planar grain boundary. After this reduction we discuss the case of compatibility between exact gradients $\bA,\bB$ on either side, and then using a symbolic computation treat the case of compatibility of simple laminates on either side as described above  (the related code being given in the Appendix). In Section \ref{Taylor} we describe the Taylor set and its properties, explaining the lower bound in \cite{DoKir}, proving the new upper bounds, and giving an independent proof of the result of Peigney \cite{Peigney2013}. The paper ends with a discussion in Section \ref{disc} of the implications of the results, related work and possible new directions.

\section {Martensitic transformations}
\label{martensite}
\subsection{Single crystals}
\label{singlecrystal}
For a single crystal and in the absence of applied forces, the nonlinear elasticity model of martensitic  transformations is based on a free energy functional \cite{j32,j40}
\begin{equation*}
I(\by,\theta)=\int_{\Omega}\psi(D\by(\bx),\theta)\:d\bx,
\end{equation*}
where the crystal occupies the bounded Lipschitz domain  $\Omega\subset\R^3$ in the reference configuration, $\by:\Omega\to\mathbb{R}^3$ is the deformation of the crystal and $\theta$ denotes the temperature. The free energy density $\psi(\bA,\theta)$ is defined for $\bA$ in  the set $GL^+(3)$, where $GL^+(n):=\{\bA\in\R^{n\times n}: \det\bA>0\}$. We denote by $K(\theta)$ the set of deformation gradients $\bA$ that minimize $\psi$ at temperature $\theta$
\begin{equation}
\label{K}
K(\theta):=\{\bA\in GL^+(3): \psi(\bA,\theta)=\min_{\bB\in GL^+(3)}\psi(\bB,\theta)\}.
\end{equation}
By adding a function of $\theta$ to $\psi$ we can and will without loss of generality  assume that $\psi(\bA,\theta)=0$ for $\bA\in K(\theta)$.  We suppose that the crystal undergoes a martensitic transformation at temperature $\theta_c$ from a cubic austenite phase for $\theta\geq\theta_c$, and take the reference configuration to be the undistorted cubic phase at $\theta=\theta_c$. For $\theta\leq \theta_c$ the change of shape of the crystal lattice is given by $\bU(\theta)=\bU(\theta)^T>\bo$, and
we can write $K(\theta)$ as the finite union of  energy wells \cite{ericksen80}:
\begin{equation}
\label{Kunion}
K(\theta)=\bigcup^N_{i=1}SO(3)\bU_i(\theta),
\end{equation}
where the martensitic variants $\bU_i(\theta)$ are  the distinct matrices $\bQ\bU(\theta)\bQ^T$ for $\bQ\in P^{24}$, and  $P^{24}$ is the  point group of the austenite consisting of the 24 rotations mapping the cube $(-1,1)^3$ to itself. The number $N$ of the martensitic variants and the corresponding transformation strains $\bU_i(\theta)$ depend on the change of symmetry for the alloy under consideration; for example, for cubic-to-tetragonal transformations (e.g. MnCu, MnNi) $N=3$,  for cubic-to-orthorhombic transformations  (e.g. CuAlNi) $N=6$ and for cubic-to-monoclinic transformations (e.g. NiTi) $N=12$. 
{\it For the rest of our discussion we fix the temperature $\theta<\theta_c$, and write $\bU_i(\theta)=\bU_i$ and $K(\theta)=K$.} 

Planar zero-energy martensitic interfaces  correspond to rank-one connections in $K$, namely to distinct matrices $\bA,\bB\in K$ which satisfy the  Hadamard jump condition 
\begin{equation}
\label{Had}
\bA-\bB=\ba\otimes\bn,
\end{equation}
for some nonzero $\ba$ and for $\bn$ the unit normal to the interface.

The existence of rank-one connections in $K$ follows from the following results.  
\begin{theorem} \cite{j48,j32,ericksen85,gurtin83}
\label{r1}
Let $\bU=\bU^T>0,$ $\bV=\bV^T>0.$ The energy wells $SO(3)\bU$ and $SO(3)\bV$ are rank-one connected iff
\[\bU^2-\bV^2=c(\bn\otimes\tilde\bn+\tilde\bn\otimes\bn),\:\:\text{for unit vectors $\bn,\tilde\bn$ and } c\neq0\]
If $\bn\nparallel \tilde\bn$ there are exactly two rank-one connections between $\bV$ and $SO(3)\bU:$
\[\bR \bU=\bV+\ba\otimes\bn,\: \tilde{\bR}\bU=\bV+\tilde{\ba}\otimes\tilde\bn,\quad \bR,\tilde{\bR}\in SO(3),\: \ba,\tilde{\ba}\in\mathbb{R}^3.\]
\end{theorem}
\begin{corollary}
\label{Cor}
Let $\bU=\bU^T>0,$ $\bV=\bV^T>0$, $\bU\neq\bV$ and $\tr \bU^2 =\tr\bV^2$.  Then $SO(3)\bU$ and $SO(3)\bV$ are rank-one connected iff
\begin{equation}\label{det.cond1}
\det(\bU^2-\bV^2)=0,
\end{equation}
and the possible interface normals $\bn,\tilde\bn$ are given up to sign by
$$\bn=\frac{\be+\bar\be}{\sqrt 2},\;\;\tilde\bn=\frac{\be-\bar\be}{\sqrt 2},$$
where $\be,\tilde\be$ are orthonormal eigenvectors of $\bU^2-\bV^2$ corresponding to its nonzero eigenvalues.
\end{corollary}
\begin{proof}
By  Theorem \ref{r1},  \eqref{det.cond1} is necessary for a rank-one connection. Conversely, if \eqref{det.cond1} holds, then since $\tr (\bU^2-\bV^2)=0$ and $\bU\neq\bV$
$$\bU^2-\bV^2=\lambda(\be\otimes\be-\bar\be\otimes\bar\be)=\lambda\left(\frac{\be+\bar\be}{\sqrt 2}\otimes \frac{\be-\bar\be}{\sqrt 2}+\frac{\be-\bar\be}{\sqrt 2}\otimes \frac{\be+\bar\be}{\sqrt 2}\right)$$
for some $\lambda\neq 0$, giving the sufficiency and normals by Theorem \ref{r1}.
\end{proof}
\begin{remark}
\label{norankone}\rm 
In particular we recover from Theorem \ref{r1} the well-known result  that there are no rank-one connections between matrices $\bA,\bB$ belonging to the same energy well $SO(3)\bU$.
\end{remark}
We recall the definition (see \cite{sverak95}) of the {\it quasiconvexification} $E^{\rm qc}$ of a compact set $E\subset\R^{m\times n}$
$$E^{\rm qc}:=\{\bA\in\R^{m\times n}:g(\bA)\leq \sup_Eg\text{ for all quasiconvex }g:\R^{m\times n}\to\R\}.$$
For the definition and properties of quasiconvex functions see, for example, \cite{dacorogna89,muller99,rindler2018}.  An alternative characterization of $E^{\rm qc}$ is 
\begin{align*}
E^{qc}&=\{\bar{\bnu}:\bnu \text{ is a homogeneous gradient Young measure with supp } \bnu\subset E\} ,
\end{align*}
where $\bar\nu:=\int_{\R^{m\times n}}\bA\,d\nu(\bA)$ denotes the centre of mass of $\nu$. For the definition and properties of  Young measures and gradient Young measures (which we abbreviate by GYM) see, for example, \cite{p21, kinderlehrerpedregal91, kinderlehrerpedregal94,muller99,rindler2018}. 
For a single crystal, zero energy microstructures correspond to a GYM $(\nu_{\bx})_{\bx\in \Omega}$ such that $\text{supp}\:\nu_{\bx} \subset K$ for a.e. $\bx\in \Omega$. The corresponding macroscopic deformation gradients $D\by(\bx)=\bar\nu_{\bx}$ belong for a.e.  $\bx\in \Omega$  to  $K^{\rm qc}$.

We will be particularly interested in homogeneous GYMs $\nu$ that are {\it laminates}, i.e. that satisfy $g(\bar\nu)\leq \langle \nu,g\rangle:=\int_{\R^{m\times n}}g(\bA)\,d\nu(\bA)$ for all rank-one convex functions $g:\R^{m\times n}\to\R$. The simplest nontrivial laminates are {\it first-order} (or {\it simple}) {\it laminates} $\nu=\lambda \delta_{\bA}+(1-\lambda)\delta_{\bB}$, where $\lambda\in[0,1]$ and $\bA-\bB=\ba\otimes\bn$, in which the compatible gradients $\bA, \bB$ are layered infinitely finely with layer normals $\bn$ and respective volume fractions $\lambda,1-\lambda$, and for which $\bar\nu=\lambda\bA+(1-\lambda)\bB$.

\subsection{Polycrystals}
\label{polycrystals}
Following \cite{p35} we consider a polycrystal which is a union of a finite number $M$ of single crystal grains, so that in the reference configuration it occupies the bounded Lipschitz domain $\Omega=\mathrm{int}\bigcup^M_{j=1}\bar\Omega_j$, where  the $\Omega_j$ are disjoint bounded Lipschitz domains. We suppose that in the grain $\Omega_j$ the crystal axes are rotated by the constant rotation $\bR_j\in SO(3)$ with respect to the cubic axes at $\theta_c$. Thus (since $\partial\Omega_j$ has zero measure for each $j$) the  total free energy for the polycrystal is given by
\begin{equation}
\label{FREEENERGYpoly}
I(\by)=\int_{\Omega}\psi(D\by(\bx)\bR(\bx))\:d\bx,
\end{equation}
where $\bR(\bx)=\bR_j$ for  $\bx\in\Omega_j$. Hence a zero energy deformation corresponds to $\by$ with $D\by(\bx) \in K\bR_j^T$ for a.e. $\bx\in \Omega_j$, and a zero energy microstructure to a GYM $(\nu_{\bx})_{\bx\in \Omega}$ such that $\text{supp}\:\bnu_x \subset K\bR_j^T$ for a.e. $\bx\in \Omega_j$.  Again, the corresponding macroscopic deformations $\by$ are those whose deformation gradients are such that $D\by(\bx)\bR_j\in K^{\rm qc}$ for a.e. $\bx\in \Omega_j$.

\section{Compatibility of laminates at  grain boundaries}\label{3}
\subsection{Reduction to the Hadamard jump condition} \label{3.1}
Consider a point $\bar\bx$ lying on the common boundary $\partial\Omega_i\cap\partial\Omega_j$ between two grains $\Omega_i,\Omega_j, i\neq j$. We suppose that this common boundary is a $C^1$ surface $\mathcal S$ in a neighbourhood of $\bar\bx$ having unit outward normal $\bm$ with respect to $\Omega_i$  at $\bar\bx$, so that for some $\delta>0$ the open ball $B:=B(\bar\bx,\delta)$ is such that $B=(B\cap\Omega_i)\cup (B\cap {\mathcal S})\cup(B\cap\Omega_j)$.

We ask when it is possible   for   a  Lipschitz zero energy deformation $\by$  to be such   that $D\by(\bx)\to \bA$ as $\bx\to\bar\bx$ with $\bx\in \om_i$, and $D\by(\bx)\to \bB$ as $\bx\to\bar\bx$ with $\bx\in\om_j$, for constant matrices $\bA,\bB$, in the sense that
\begin{align}
\label{limits}
&\lim_{\ep\to 0+}\av_{B(\bar\bx,\ep)}|D\by(\bx)-\bA|\,\chi_i(\bx)\,d\bx=0,\;\;\lim_{\ep\to 0+}\av_{B(\bar\bx,\ep)}|D\by(\bx)-\bB|\,\chi_j(\bx)\,d\bx=0,
\end{align}
where $\chi_i,\chi_j$ denote the characteristic functions of $\om_i,\om_j$ respectively and $\displaystyle\av$ denotes average. Blowing up around $\bar\bx$ by setting
$\bw_\ep(\bz)=\ep^{-1}(\by(\bar\bx+\ep\bz)-\by(\bar\bx))$, $\ep>0$, we see that \eqref{limits} is equivalent to 
\begin{align}
\label{limitsblownup}
&\lim_{\ep\to 0+}\av_{B(\bzero,1)}|D\bw_\ep(\bz)-\bA|\,\chi_i(\bar\bx+\ep\bz)\,d\bz=0,\;\lim_{\ep\to 0+}\av_{B(\bzero,1)}|D\bw_\ep(\bz)-\bB|\,\chi_j(\bar\bx+\ep\bz)\,d\bz=0.
\end{align}
 Since $\bw_\ep(\bzero)=\bzero$ and $D\bw_\ep(\bz)$ is bounded (because $D\by(\bx)\in KSO(3)$), it follows that $\bw_\ep$ is bounded in $W^{1,\infty}(B(\bzero,1),\R^3)$, so that $\bw_{\ep_r}\weakstar\bw$ in  $W^{1,\infty}(B(\bzero,1),\R^3)$ for some $\bw$ and sequence $\ep_r\to 0$. By the compactness of the embedding of $W^{1,\infty}(B(\bzero,1),\R^3)$ in $C^0(\overline{B(\bzero,1)},\R^3)$, $\bw_{\ep_r}\to \bw$ in $C^0(\overline{B(\bzero,1)},\R^3)$. Since for any $\delta>0$ and sufficiently large $j$ we have  $\chi_i(\bar\bx+\ep_r\bz)=1$ for $\bz\in B(\bzero,1)$ with $\bz\cdot\bm<-\delta$ we obtain from \eqref{limitsblownup} that $D\bw_{\ep_r}\to \bA$ strongly in $L^1(B(\bzero,1)\cap \{\bz\cdot\bm<-\delta\},\R^{3\times 3})$, so that $D\bw(\bz)=\bA$ for $\bz\cdot\bm<0$. Similarly $D\bw(\bz)=\bB$ for $\bz\cdot\bm>0$. Therefore for $\bx\in B(\bzero,1)$
\begin{equation}
\label{AB}
\bw(\bx)=\left\{\begin{array}{ll}\bA\bx,& \bx\cdot\bm\leq \bo,\\\bB\bx,& \bx\cdot\bm\geq \bo.\end{array}\right.
\end{equation}
Since $D\by(\bx)\in K\bR_i^T$ for a.e. $\bx\in\Omega_i$ the above strong convergence implies that $\bA\in K\bR_i^T$, and similarly $\bB\in K\bR_j^T$, so that
\begin{equation}
\label{AB1}
\bA\in SO(3)\bR_i\bU_k\bR_i^T,\; \bB\in SO(3)\bR_j\bU_l\bR_j^T,
\end{equation}
for some $k,l\in\{1,\ldots,N\}$.  Thus we have reduced  the problem to that of a planar grain boundary separating zero-energy constant gradients.

More generally, let $(\nu_{\bx})_{\bx\in B}$ be a zero energy GYM, so that $\supp\nu_{\bx}\subset K\bR_i^T$ for a.e. $\bx\in B\cap \om_i$, $\supp\nu_{\bx}\subset K\bR_j^T$ for a.e. $\bx\in B\cap\om_j$, satisfying 
\begin{align}
\label{limitsnu}
&\lim_{\ep\to 0+}\av_{B(\bar\bx,\ep)}d(\nu_{\bx},\nu_i)\chi_i(\bx)\,d\bx=0,\; \lim_{\ep\to 0+}\av_{B(\bar\bx,\ep)}d(\nu_{\bx},\nu_j)\chi_j(\bx)\,d\bx=0,
\end{align}
where $\nu_i,\nu_j$ are probability measures on $K$, and $d$ is a distance on the set $P(KSO(3))$ of probability measures on $KSO(3)$ such that $\mu_j\weakstar\mu$ in $P(KSO(3))$ iff $d(\mu_j,\mu)\to 0$ (possible choices include the Prohorov metric \cite[p72ff]{billingsley} and the Wasserstein-1 metric \cite[Chapter 6]{villanioptimal}). Setting $\nu^{\ep}_{\bz}=\nu_{\bar\bx+\ep\bz}$ \eqref{limitsnu} is equivalent to
\begin{align}
\label{limitsnublownup}
\lim_{\ep\to 0+}\av_{B(\bzero,1)}d(\nu^{\ep}_{\bz},\nu_i)\chi_i(\bar\bx+\ep\bz)\,d\bz=0,\;\lim_{\ep\to 0+}\av_{B(\bzero,1)}d(\nu^{\ep}_{\bz},\nu_j)\chi_j(\bar\bx+\ep\bz)\,d\bz=0.
\end{align}
 We claim that
\begin{equation}
\label{nublowup}
\mu_{\bz}:=\left\{\begin{array}{cc}\nu_i,&\bz\cdot\bm<0,\\ \nu_j,&\bz\cdot\bm>0,\end{array}\right.
\end{equation}
is a GYM on $B(\bzero,1)$. To prove this first note that $\bar\nu_{\bx}=D\by(\bx)$ for  a.e. $\bx\in B$ for some $\by\in W^{1,\infty}(B,\R^3)$. From \eqref{limitsnublownup} we see that $d(\nu^{\ep}_{\bz},\nu_i)\to 0$ in $L^1(B(\bzero,1)\cap\{\bz\cdot\bm<-\delta\})$ for  $\delta>0$, so that there exists a sequence $\ep_r\to 0$ such that $\nu^{\ep_r}_{\bz}\weakstar \nu_i$ for a.e. $\bz\in B(\bzero,1)\cap \{\bz\cdot\bm<0\}$ and similarly such that $\nu^{\ep_r}_{\bz}\weakstar \nu_j$ for a.e. $\bz\in B(\bzero,1)\cap \{\bz\cdot\bm>0\}$, so that in particular $\bar\nu^{\ep_r}_{\bz}\to \bar\mu_{\bz}$ for a.e. $\bz\in B(\bzero,1)$.
Hence
 for any quasiconvex function $\varphi:\R^{3\times 3}\to\R$ we obtain for such a $\bz$ with $\bz\cdot\bm<0$
$$\langle\nu_i,\varphi\rangle=\lim_{r\to\infty}\langle\nu^{\ep_r}_{\bz},\varphi\rangle\geq\lim_{r\to\infty}\varphi(\bar\nu_{\bz}^{\ep_r})= \varphi(\bar\nu_i),$$
and similarly $\langle\nu_j,\varphi\rangle\geq \varphi(\bar\nu_j)$. Furthermore,  as $D\by(\bar\bx+\ep_r\bz)=\bar\nu^{\ep_r}_{\bz}$ is bounded in $L^\infty(B(\bzero,1),\R^{3\times 3})$ we can suppose also that $D\by(\bar\bx+\ep_r\bz)$ converges weak$\,*$ in $L^\infty(B(\bzero,1),\R^{3\times 3})$, and the weak limit must equal the a.e. limit $\bar\mu_{\bz}$.  But the weak$\,*$ limit of a sequence of gradients is a gradient, so that $\bar\mu_{\bz}$ is a gradient.
  The claim then  follows from the characterization of GYM due to Kinderlehrer \& Pedregal \cite[Theorem 6.1]{kinderlehrerpedregal91}.
\subsection{Zero order laminates: the case of exact gradients}\label{3.2}
The Hadamard jump condition gives that \eqref{AB} holds iff
$\bA-\bB=\ba\otimes \bm$ for some $\ba\in\R^3$. There are now two possibilities (i) that $\bR_i\bU_k\bR_i^T=\bR_j\bU_l\bR_j^T$, (ii) that $\bR_i\bU_k\bR_i^T\neq \bR_j\bU_l\bR_j^T$, so that there is a rank-one connection between $SO(3)\bR_i\bU_k\bR_i^T$ and $SO(3)\bR_j\bU_l\bR_j^T$.

Let us first consider the case (i), when by Remark  \ref{norankone} we have $\bA=\bB$, so that there is no jump in the deformation gradient across the interface. This case holds iff
\begin{equation}
\label{Ezero}
\bU_k=\bR\bU_l\bR^T,
\end{equation}
where $\bR:=\bR_i^T\bR_j$ is the relative rotation of the cubic axes in the grains. Let $\bU_l$ have spectral decomposition
 \begin{equation}
\label{spectral}\bU_l=\sum_{r=1}^3\lambda_r\hat \be_r\otimes\hat\be_r.
\end{equation}
By the definition of the variants, $\bU_k=\bQ\bU_l\bQ^T$ for some $\bQ\in P^{24}$. Hence  
$\bU_k=\bR\bU_l\bR^T$ iff 
\begin{equation}
\label{Rtilde}
\tilde\bR\bU_l\tilde\bR^T=\bU_l,
\end{equation}
 where $\tilde\bR:=\bR^T\bQ$. 
 If $\bR\in P^{24}$ then taking $\bQ=\bR$ we see that \eqref{Ezero} can hold for suitable $k,l$. However \eqref{Ezero} can hold for $\bR\in SO(3)\setminus P^{24}$, as shown by the following result.

\begin{proposition}
\label{prop1}
$\bU_k=\bR \bU_l\bR ^T$ for some $\bR\in SO(3)\setminus P^{24}$ iff either at least two  $\lambda_r$ are equal, or the $\lambda_r$ are distinct and $-\bI+2\hat\be_s\otimes\hat\be_s\not\in P^{24}$ for some $s$. 
\end{proposition}
\begin{proof}
 \eqref{Rtilde} holds iff
\begin{align}
\label{eigen}
\sum_{r=1}^{3}\lambda_r\tilde \bR \hat{\be}_r\otimes \tilde\bR \hat{\be}_r=\sum_{r=1}^{3}\lambda_r \hat{\be}_r\otimes \hat{\be}_r
\end{align}
If two of the eigenvalues are equal, say $\lambda_2=\lambda_3 ,$ then 
\begin{align*}
\tilde\bR =\hat{\be}_1\otimes \hat{\be}_1 + \cos\theta\, (\hat{\be}_2\otimes \hat{\be}_2+\hat{\be}_3\otimes \hat{\be}_3) + \sin\theta\, (\hat{\be}_3\otimes \hat{\be}_2-\hat{\be}_2\otimes \hat{\be}_3)  
\end{align*}
 satisfies \eqref{eigen}, and we can choose $\theta$ such that  $\tilde\bR\notin P^{24}$.
 If all the $\lambda_r$  are distinct then \eqref{eigen} holds iff   $\tilde\bR \hat{\be}_r=\kappa_r\hat{\be}_r$  for each $r$ with $\kappa_r=\pm 1$ and $\kappa_1\kappa_2\kappa_3=1$.  If $\kappa_1=\kappa_2=\kappa_3$ then $\tilde\bR=\bI\in P^{24}$ while otherwise $\tilde\bR=-\bI+2 \hat{\be}_s\otimes \hat{\be}_s$ for some $s$.
\end{proof}
Thus for a cubic-to-tetragonal transformation, for which we can take
$$\bU_l= \eta_2\be_1\otimes\be_1+\eta_1(\be_2\otimes\be_2+\be_3\otimes\be_3),\;\eta_1>0, \eta_2>0, \eta_1\neq\eta_2,$$
 two eigenvalues are equal and so there are $\bR\in SO(3)\setminus P^{24}$ such that $\bU_k=\bR \bU_l\bR ^T$. 

 On the other hand, for a cubic-to-orthorhombic transformation with transformation strain   
$$\bU_l=\alpha \be_1\otimes \be_1+\beta\frac{\be_2+\be_3}{\sqrt{2}}\otimes \frac{\be_2+\be_3}{\sqrt{2}}+\gamma\frac{\be_2-\be_3}{\sqrt{2}}\otimes \frac{\be_2-\be_3}{\sqrt{2}},$$
and distinct positive  eigenvalues $\alpha,\beta,\gamma$, we have that
\begin{align*}
-\bI+2 \be_1\otimes \be_1 \quad \text{and} \quad -\bI+2 \frac{\be_2\pm \be_3}{\sqrt{2}}\otimes \frac{\be_2\pm \be_3}{\sqrt{2}} \in P^{24},
\end{align*}
so that by Proposition \ref{prop1} the only $\bR\in SO(3)$ with  $\bU_k=\bR \bU_l\bR ^T$ are in $P^{24}$. The same conclusion holds for a cubic-to-orthorhombic transformation having the less common transformation strain
\begin{equation}
\label{lesscommon}\bU_l=\alpha\be_1\otimes\be_1+\beta\be_2\otimes\be_2+\gamma\be_3\otimes\be_3,
\end{equation}
with distinct positive  eigenvalues $\alpha,\beta,\gamma$.

We remark that it can happen that $-\bI+2\hat\be_s\otimes\hat\be_s\not\in P^{24}$ for some $s$. This occurs, for example, in cubic-to-rhombohedral transformations such as to the $R$ phase in NiTi, when
 $$\bU=\delta\bI+(\tau-\delta)\hat\be\otimes\hat\be$$ for $\delta>0, \tau>0$ with $\tau\neq\delta$, and $\hat\be=\frac{1}{\sqrt{3}}(\be_1+\be_2+\be_3)$, since  $\hat\be$ is a  3-fold axis of the cube and so $-\bI+2\hat\be\otimes\hat\be\not\in P^{24}$. But in this case $\bU$ has two equal eigenvalues, so that we already know from Proposition \ref{prop1} that $\bU_k=\bR \bU_l\bR ^T$ for some $\bR\in SO(3)\setminus P^{24}$.

We now consider case (ii), which by Corollary \ref{Cor} holds iff
\begin{equation}
\label{detE}
\det\bE=0, \text{where } \bE:=\bU_k^2-\bR\bU_l^2\bR^T.
\end{equation}
(Note that case (i) corresponds to $\bE=\bzero$.) With the same notation $\tilde\bR=\bR^T\bQ$ as for case (i) we have that
$$\det\bE=\det (\tilde\bR\bU_l^2-\bU_l^2\tilde\bR),$$
which we can evaluate explicitly using the axis-angle representation of $\tilde\bR$ (Rodrigues' formula, see e.g. \cite{MarsdenRatiu})
\begin{align}
\label{axisangle}
\tilde{\bR}=\exp(\varphi\bN)=\cos\varphi \,\bI+\sin \varphi\,\bN + (1-\cos\varphi)\,\bn\otimes\bn,
\end{align}
where the unit vector $\bn$ is the axis of rotation, $\varphi\in[0,2\pi)$ is the angle of rotation, and $\bN$ is the skew matrix satisfying $\bN\bp=\bn\wedge\bp$ for all $\bp\in\R^3$, that is
\begin{align*}
\bN=\begin{pmatrix}
    		0 & -n_3 & n_2\\
    		n_3 &  0 & -n_1\\
    		-n_2 & n_1 & 0
    	    \end{pmatrix},
\end{align*}
where $n_r:=\bn\cdot\hat\be_r$ are the components of $\bn$ in the $\hat\be_k$ basis.
\begin{theorem}
\label{Theorem1}
Let $\bU_l$ have spectral decomposition
$\bU_l=\sum_{r=1}^{3}\lambda_r \hat{\be}_r\otimes \hat{\be}_r$
and $\tilde\bR$ have axis-angle representation \eqref{axisangle}.
 Then
\begin{align}
\label{detformula}
\det(\tilde{\bR}\bU_l^2-\bU_l^2\tilde{\bR})=4\:\rho\: n_1n_2n_3\:\sin \varphi(1-\cos \varphi), 
\end{align}
where $\rho=(\lambda_1^2-\lambda_2^2)(\lambda_2^2-\lambda_3^2)(\lambda_3^2-\lambda_1^2),$  so that
\begin{equation}
\label{DETCondition}
\det(\tilde{\bR}\bU_l^2-\bU_l^2\tilde{\bR})=0
\end{equation}
iff $(i)$ two of the $\lambda_r$ are equal, or $(ii)$ $\bn\cdot\hat{\be}_r=0$ for some $r$, or $(iii)$ $\varphi=0$ or $\varphi=\pi.$ 
\end{theorem}
\begin{proof}
Observe that \eqref{axisangle} yields
\begin{align}
\label{det}
\tilde{\bR}\bU_l^2-\bU_l^2\tilde{\bR}= \sin\varphi\, (\bN \bU_l^2-\bU_l^2\bN) + (1-\cos\varphi)\, (\bn\otimes \bU_l^2\bn-\bU_l^2 \bn\otimes\bn).
\end{align}
We calculate
\begin{align*}
\bN \bU_l^2-\bU_l^2 \bN 
  &=\begin{pmatrix}
  	0 & (\lambda_1^2-\lambda_2^2)n_3 & (\lambda_3^2-\lambda_1^2)n_2\\
  	(\lambda_1^2-\lambda_2^2)n_3 &  0 & (\lambda_2^2-\lambda_3^2)n_1\\
  	(\lambda_3^2-\lambda_1^2)n_2 & (\lambda_2^2-\lambda_3^2)n_1 & 0
  \end{pmatrix},
\end{align*}
\begin{align*}
\bn\otimes \bU_l^2 \bn- \bU_l^2\bn \otimes \bn&=
\begin{pmatrix}
0 & (\lambda_2^2-\lambda_1^2)n_1n_2 & (\lambda_3^2-\lambda_1^2)n_1n_3\\
(\lambda_1^2-\lambda_2^2)n_1n_2 &  0 & (\lambda_3^2-\lambda_2^2)n_2n_3\\
(\lambda_1^2-\lambda_3^2)n_1n_3 & (\lambda_2^2-\lambda_3^2)n_2n_3 & 0
\end{pmatrix}.
\end{align*}
Therefore from \eqref{det} we obtain
\begin{align*}
\tilde{\bR}\bU_l^2-\bU_l^2\tilde{\bR}=
\begin{pmatrix}
   		0 & p_{12} & p_{13}\\
   		p_{21} &  0 & p_{23}\\
   		p_{31} & p_{32} & 0
\end{pmatrix}
\end{align*}
where
     \begin{align*}
	 &p_{12}=(\lambda_1^2-\lambda_2^2)(n_3\sin \varphi-(1-\cos \varphi)n_1n_2),
	 &p_{13}=(\lambda_3^2-\lambda_1^2)(n_2\sin \varphi+(1-\cos \varphi)n_1n_3),\\
	 &p_{21}=(\lambda_1^2-\lambda_2^2)(n_3\sin \varphi+(1-\cos \varphi)n_1n_2),
	 &p_{23}=(\lambda_2^2-\lambda_3^2)(n_1\sin \varphi-(1-\cos \varphi)n_2n_3),\\
	 &p_{31}=(\lambda_3^2-\lambda_1^2)(n_2\sin \varphi-(1-\cos \varphi)n_1n_3),
	 &p_{32}=(\lambda_2^2-\lambda_3^2)(n_1\sin \varphi+(1-\cos \varphi)n_2n_3),
    \end{align*}
from which a short calculation gives
\begin{align*}
\det(\tilde{\bR}\bU_l^2-\bU_l^2\tilde{\bR})=p_{12}p_{23}p_{31}+p_{13}p_{32}p_{21}=4\rho n_1n_2n_3\sin\varphi\,(1-\cos\varphi)
\end{align*}
as claimed. 
\end{proof}
\begin{corollary}
\label{Cor2}
Let $\bU_l$ (and hence all $\bU_k$) have three distinct eigenvalues. Then for a dense open set of relative  rotations $\bR_i^T \bR_j\in SO(3)$ there is no piecewise affine deformation of the form \eqref{AB} with $\bA\in SO(3)\bR_i\bU_k\bR_i^T,\; \bB\in SO(3)\bR_j\bU_l\bR_j^T$ for any $k,l$ and normal $\bm$. 
\end{corollary}
\begin{proof} Fix $l=1$ with $\hat\be_r$ the basis of eigenvectors of   $\bU_1$ and define
$$S:=\{\tilde\bR=\exp(\varphi\bN):n_1n_2n_3\sin\varphi\,(1-\cos\varphi)\neq 0\},$$
where as before $n_r=\bn\cdot\hat\be_r$.  Then $S$
is an open dense subset of $SO(3)$, and hence so is $\bar S:=\cap\{\bQ_l S\bQ_k^T:\bQ_k,\bQ_l\in P^{24}\}$. Given $k,l$ we have that $\bU_k=\bQ_k\bU_1\bQ_k^T$, $\bU_l=\bQ_l\bU_1\bQ_l^T$ for some $\bQ_k,\bQ_l\in P^{24}$. Thus,
\begin{align}\label{det1}
\det (\bR_i\bU_k^2\bR_i^T-\bR_j\bU_l^2\bR_j^T)&=\det(\bR_i\bQ_k\bU_1^2\bQ_k^T\bR_i^T-\bR_j\bQ_l\bU_1^2\bQ_l^T\bR_j^T)\\&=\det(\bQ_l^T\bR_j^T\bR_i\bQ_k\bU_1^2-\bU_1^2\bQ_l^T\bR_j^T\bR_i\bQ_k),\nonumber
\end{align}
which is nonzero iff $\bQ_l^T\bR_j^T\bR_i\bQ_k\in S$. Hence if $\bR_j^T\bR_i\in \bar S$ (equivalently $\bR_i^T\bR_j\in\bar S$, since $\bar S^T=\bar S$) the determinant in \eqref{det1} is nonzero for any $k,l$ and so there is no piecewise affine deformation of the required form.
\end{proof}
If $\det \bE=0$ (in particular if $\bU_l$ has two equal eigenvalues) and $\bE\neq \bo$, then by Corollary \ref{Cor} there are exactly two possible corresponding grain boundary normals $\bm$ consistent with \eqref{AB}, and they are orthogonal and $(101)$ directions with respect to the basis of eigenvectors of $\bR_i\bU_k^2\bR_i^T-\bR_j\bU_l^2\bR_j^T$.

 The results of this section show that, for general relative grain rotations and grain boundary normals, a zero energy deformation cannot satisfy \eqref{limits}, so that a nontrivial microstructure is necessary to achieve compatibility across the grain boundary. 

\subsection{First order laminates}\label{3.3}
As  discussed in Section \ref{3.2} constant deformation gradients are not enough to achieve compatibility across a grain boundary for general grain boundary normals and relative grain rotations. Thus it is natural to ask whether compatibility can in general be achieved by replacing constant gradients by laminates.  In this section we investigate whether we can have compatibility with a zero energy  microstructure that tends to a simple laminate on each side of a  point of a locally $C^1$ grain boundary in the sense of \eqref{limitsnu}.  The analysis in Section \ref{3.1} shows that we may then assume that the grain boundary is planar with a simple laminate on each side. We  show that in general it is not possible to achieve compatibility with such a microstructure.  As a result, higher order laminates are required in order to achieve kinematic compatibility. 

To understand why this is the case it is useful to count equations and variables. The Hadamard jump condition for compatibility across a planar grain boundary with normal $\bm$ consists of an equation of the form
$$\bF-\bG=\bb\otimes \bm$$
where $\bF$ and $\bG$ are the macroscopic deformation gradients of the two simple laminates, with $\bG$ depending on  the relative rotation $\bR$ of the grains. In $\bF$ and $\bG$ there are effectively 5 unknowns consisting of two volume fractions and a rotation, which together with the three unknowns in $\bb$ gives a total of 8 unknowns for 9 equations, so that generically (that is for almost  all $\bm$ and $\bR$) we don't expect a solution. This claim is made rigorous below.

We consider then a planar grain boundary  $\bx\cdot\bm=k$ separating two grains $\om_1:=\{\bx\in B(\bzero,1):\bx\cdot\bm<0\}$, $\om_2:=\{\bx\in B(\bzero,1):\bx\cdot\bm>0\}$ with relative rotation $\bR$. Without loss of generality we can assume that the cubic axes in $\om_1$ are unrotated, so that  a first order zero energy laminate in $\om_1$ is represented by the GYM
\begin{equation}
\label{lam1}
\bnu_1=
\lambda\delta_{\bA}+(1-\lambda)\delta_{\bB}, \; \; \bA,\bB\in K,\;\;\bA-\bB=\hat\ba_i\otimes \bn_i, \;\; \lambda\in[0,1].
\end{equation}
Similarly,  a first order zero energy laminate in $\om_2$ is represented by the GYM
\begin{equation}
\label{lam2}
\bnu_2=
\mu\delta_{\bC\bR^T}+(1-\mu)\delta_{\bD\bR^T},\;\;\bC,\bD\in K,\;\; \bC-\bD=\hat\ba_j\otimes \bn_j,\;\;\mu\in[0,1].
\end{equation}
  In order that these two simple laminates be compatible across the grain boundary, that is for 
$$\bnu_{\bx}:=\left\{\begin{array}{cc}\bnu_1&\bx\in\om_1,\\\bnu_2&\bx\in\om_2,\end{array}\right.$$
to be a GYM on  $B(\bzero,1)$, it is necessary and sufficient by \cite[Theorem 6.1]{kinderlehrerpedregal91} that
\[\bar\bnu_{1}-\bar\bnu_{2}=\hat\bb\otimes\bm\quad \text{for some } \hat\bb\neq\bo,\]
or equivalently that the corresponding macroscopic deformation gradients satisfy
\begin{align}
\label{BA}
\lambda\bA+(1-\lambda)\bB-(\mu\bC+(1-\mu)\bD)\bR^T=\hat\bb\otimes\bm.
\end{align}
Since $\bA,\bB,\bC,\bD\in K$ we have that $\bA=\bQ_1\bU_k,$ $\bB=\bQ_2\bU_i,$ $\bC=\bQ_3\bU_l,$ $\bD=\bQ_4\bU_j,$ for some $i,j,k,l$. We then substitute $\bQ_1\bU_k-\bQ_2\bU_i=\hat\ba_i\otimes \bn_i,$ $\bQ_3\bU_l-\bQ_4\bU_j=\hat\ba_j\otimes\bn_j$ into \eqref{BA} to get
\begin{align}
\label{gbcompat}
\bU_i+\lambda\ba_i\otimes\bn_i-\bQ(\bU_j+\mu\ba_j\otimes\bn_j)\bR^T
=\bb\otimes\bm,
\end{align}
where $\bQ:=\bQ_2^T\bQ_4$, $\ba_i:=\bQ_2^T\hat\ba_i$, $\ba_j:=\bQ_4^T\hat\ba_j$ and $\bb:=\bQ_2^T\hat\bb$. Thus $\bU_i+\ba_i\otimes\bn_i=\bQ_2^T\bQ_1\bU_k$ and $\bU_j+\ba_j\otimes\bn_j=\bQ_4^T\bQ_3\bU_l$ give rank-one connections from $\bU_i$ to $SO(3)\bU_k$ and from $\bU_j$ to $SO(3)\bU_l$ respectively.

Let $\bP=\boldsymbol{1}-\bm\otimes\bm$ denote the orthogonal projection of $\R^3$ onto the plane $\bm^{\perp}=\{\bx\in \mathbb{R}^3:\:\bm\cdot\bx=\bo\}$ orthogonal to $\bm$. Then
\begin{equation}
\nonumber
(\bU_i+\lambda\ba_i\otimes\bn_i)\bP=\bQ(\bU_j+\mu\ba_j\otimes\bn_j)\bR^T\bP,
\end{equation}
and so
\begin{equation}
\label{PSP}
\bP\bS\bP=\bzero,
\end{equation}
where
\begin{align*}
\bS(\lambda,\mu)&:=
(\bU_i+\lambda\bn_i\otimes\ba_i)(\bU_i+\lambda\ba_i\otimes\bn_i)-\bR(\bU_j+\mu\bn_j\otimes\ba_j)(\bU_j+\mu\ba_j\otimes\bn_j)\bR^T\\
&=\bU_i^2+\lambda(\bU_i\ba_i\otimes\bn_i+\bn_i\otimes\bU_i\ba_i)
+\lambda^2|\ba_i|^2\bn_i\otimes\bn_i\\
&\qquad
-\bR\bU_j^2\bR^T-\mu\bR(\bU_j\ba_j\otimes\bn_j+\bn_j\otimes\bU_j\ba_j)\bR^T-\mu^2|\ba_j|^2\bR\bn_j\otimes\bR\bn_j.
\end{align*}
Let $\{\bm_1,\bm_2\}$ be an orthonormal basis of $\bm^{\perp}$. Then \eqref{PSP} is equivalent to the equation in $\R^3$
$$(\bS\bm_ 1\cdot\bm_1,\bS\bm_1\cdot\bm_2,\bS\bm_2\cdot\bm_2)=\bzero,$$
that is to the equation
\begin{align}
\label{ElmA}
\bE(\lambda,\mu):=\bal_0+\lambda\bal_1+\lambda^2\bal_2-\mu\bet_1-\mu^2\bet_2=\bo,
\end{align}
for vectors $\bal_0, \bal_1,\bal_2,\bet_1,\bet_2\in\R^3$ given by
\begin{align*}\bal_0&=(\bA_0\bm_1\cdot\bm_1, \bA_0\bm_1\cdot\bm_2, \bA_0\bm_2\cdot \bm_2),\\
\bal_1&=(\bA_1\bm_1\cdot\bm_1, \bA_1\bm_1\cdot\bm_2, \bA_1\bm_2\cdot \bm_2),\\
\bal_2&=(\bA_2\bm_1\cdot\bm_1, \bA_2\bm_1\cdot\bm_2, \bA_2\bm_2\cdot \bm_2),\\
\bet_1&=(\bB_1\bm_1\cdot\bm_1, \bB_1\bm_1\cdot\bm_2, \bB_1\bm_2\cdot \bm_2),\\
\bet_2&=(\bB_2\bm_1\cdot\bm_1, \bB_2\bm_1\cdot\bm_2, \bB_2\bm_2\cdot \bm_2),
\end{align*}
where
\begin{align}
\label{Melements}
\begin{split}
\bA_0:=\bU_i^2-\bR\bU_j^2\bR^T,\quad
\bA_1:=\bU_i\ba_i\otimes\bn_i+\bn_i\otimes\bU_i\ba_i,\quad
\bA_2:=|\ba_i|^2\bn_i\otimes\bn_i,\\
\bB_1:=\bR\bU_j\ba_j\otimes\bn_j+\bn_j\otimes\bR\bU_j\ba_j,\quad
\bB_2:=|\ba_j|^2\bR\bn_j\otimes\bR\bn_j.     
\end{split}
\end{align}
Thus we have reduced the problem of solving the nine equations \eqref{gbcompat} in the eight unknowns $\lambda,\mu, \bQ, \bb$ to that of solving the three equations \eqref{ElmA} in the two unknowns $\lambda,\mu$.

A necessary condition for a solution to \eqref{ElmA} to exist with $\lambda,\mu\in[0,1]$  is the vanishing of the Dixon resultant $\det\bcM$ (see \cite{KapurSaxenaYang}), which is a polynomial in $\bal_0,\bal_1,\bal_2,\bet_1,\bet_2$ defined in terms of the scalar triple product 
\begin{align*}
\Delta(\lambda,\mu,\bar\lambda,\bar\mu)
:&=\det(\bE(\lambda,\mu),\bE(\bar\lambda,\mu),\bE(\bar\lambda,\bar\mu))\\
&=[\bE(\lambda,\mu),\bE(\bar\lambda,\mu),\bE(\bar\lambda,\bar\mu)]\\
&=[\bE(\lambda,\mu)-\bE(\bar\lambda,\mu),\bE(\bar\lambda,\mu)-\bE(\bar\lambda,\bar\mu),\bE(\bar\lambda,\bar\mu)].
\end{align*}
Since
\begin{align*}
\frac{\bE(\lambda,\mu)-\bE(\bar\lambda,\mu)}{\lambda-\bar\lambda}
=\bal_2(\lambda+\bar\lambda)+\bal_1, \qquad
\frac{\bE(\bar\lambda,\mu)-\bE(\bar\lambda,\bar\mu)}{\mu-\bar\mu}
=-\bet_2(\mu+\bar\mu)-\bet_1,
\end{align*}
we have that
\begin{align*}
\frac{\Delta(\lambda,\mu,\bar\lambda,\bar\mu)}{(\lambda-\bar\lambda)(\mu-\bar\mu)}
=[\bal_2(\lambda+\bar\lambda)+\bal_1, -\bet_2(\mu+\bar\mu)-\bet_1, \bE(\bar\lambda,\bar\mu)].
\end{align*}
Suppose that for some $(\lambda_0,\mu_0)\in [0,1]^2$ we have $\bE(\lambda_0,\mu_0)=\bo.$ Then $(\lambda_0-\bar\lambda)(\mu_0-\bar\mu)$ is a nonzero polynomial in $\bar\lambda$ and $\bar\mu.$ Since $\Delta(\lambda_0,\mu_0,\bar\lambda,\bar\mu)=0$ for all $\bar\lambda$ and $\bar\mu,$ and polynomials in several variables form an integral domain (see, for example, \cite[p305]{dummitfoote}), it follows that 
\[\frac{\Delta(\lambda_0,\mu_0,\bar\lambda,\bar\mu)}{(\lambda_0-\bar\lambda)(\mu_0-\bar\mu)}
=0 \quad\text{for all $\bar\lambda$ and $\bar\mu.$}\] 
This is a polynomial in $\bar\lambda$ and $\bar\mu$ and thus all its coefficients must be zero. Now
\begin{align*}
\frac{\Delta(\lambda_0,\mu_0,\bar\lambda,\bar\mu)}{(\lambda_0-\bar\lambda)(\mu_0-\bar\mu)}
&=[\bal_2(\lambda_0+\bar\lambda)+\bal_1, -\bet_2(\mu_0+\bar\mu)-\bet_1, \bal_0+\bar\lambda\bal_1+\bar\lambda^2\bal_2-\bar\mu\bet_1-\bar\mu^2\bet_2]\\
&=(1,\bar\lambda,\bar\mu,\bar\lambda\bar\mu)\bcM(1,\lambda_0,\mu_0,\lambda_0\mu_0)^T,
\end{align*}
where
\[\bcM=
\begin{pmatrix}
[\bal_1,-\bet_1,\bal_0] & [\bal_2,-\bet_1,\bal_0] & [\bal_1,-\bet_2,\bal_0] & [\bal_2,-\bet_2,\bal_0] \\
[\bal_2,-\bet_1,\bal_0] & [\bal_2,-\bet_1,\bal_1] & [\bal_2,-\bet_2,\bal_0] & [\bal_2,-\bet_2,\bal_1]\\
[\bal_1,-\bet_2,\bal_0] & [\bal_2,-\bet_2,\bal_0] & [\bal_1,-\bet_2,-\bet_1] & [\bal_2,-\bet_2,-\bet_1]\\
[\bal_2,-\bet_2,\bal_0] & [\bal_2,-\bet_2,\bal_1] & [\bal_2,-\bet_2,-\bet_1] & 0
\end{pmatrix}.\]
Hence, we obtain $\bcM(1,\lambda_0,\mu_0,\lambda_0\mu_0)^T=\bo$ and thus $\det\bcM=0.$

The following lemma allows us to express $\bcM$ directly as a function of $\bm$ instead of $\{\bm_1,\bm_2\}$.
\begin{lemma} \label{LemmaM}
Let
\begin{align*}
\ba&=(\bA\bm_1\cdot\bm_1,\bA\bm_2\cdot\bm_2,\bA\bm_2\cdot\bm_1)\\
\bb&=(\bB\bm_1\cdot\bm_1,\bB\bm_2\cdot\bm_2,\bB\bm_2\cdot\bm_1)\\
\bc&=(\bC\bm_1\cdot\bm_1,\bC\bm_2\cdot\bm_2,\bC\bm_2\cdot\bm_1)
\end{align*}
where $\bA,\bB,\bC\in\mathbb{R}^{3\times 3}$ are symmetric and $\{\bm_1,\bm_2,\bm\}$ is a basis for $\mathbb{R}^3,$ and let $\bM$ be the skew-symmetric matrix such that $\bM\bx:=\bm\wedge\bx,$ so that $\bM_{kj}:=\varepsilon_{ijk}\bm_k$. Then $$[\ba,\bb,\bc]=-\tr(\bA\bM\bB\bM\bC\bM).$$
\end{lemma}
\begin{proof}
In the $\{\bm_1,\bm_2,\bm\}$ basis the matrix 
\begin{equation*}
\bM=
\begin{pmatrix}
0 & -1 & 0 \\
1 & 0 & 0\\
0 & 0 & 0
\end{pmatrix}    
\end{equation*}
while the symmetric matrices $\bA,\bB,\bC\in\mathbb{R}^{3\times 3}$ are
\[\bA=
\begin{pmatrix}
a_1 & a_3 & a_4 \\
a_3 & a_2 & a_5\\
a_4 & a_5 & a_6
\end{pmatrix},
\:\bB=
\begin{pmatrix}
b_1 & b_3 & b_4 \\
b_3 & b_2 & b_5\\
b_4 & b_5 & b_6
\end{pmatrix},
\:\bC=
\begin{pmatrix}
c_1 & c_3 & c_4 \\
c_3 & c_2 & c_5\\
c_4 & c_5 & c_6
\end{pmatrix}.
\]
Then we calculate
\begin{align*}
\tr(\bA\bM\bB\bM\bC\bM)&=
c_3(a_3b_3-a_1b_2)+c_2(a_1b_3-a_3b_1)-c_1(a_2b_3-a_3b_2)-c_3(a_3b_3-a_2b_1)\\
&=-a_1b_2c_3-a_3b_1c_2+a_1b_3c_2-a_2b_3c_1+a_3b_2c_1+a_2b_1c_3\\
&=-[\ba,\bb,\bc].
\end{align*}
\end{proof}
Thus by\eqref{ElmA}, \eqref{Melements}
\begin{small}
\begin{align}
\label{M}
\bcM=
\setlength{\arraycolsep}{1pt}
\begin{pmatrix}
-\tr(\bA_1\bM\bB_1\bM\bA_0\bM) & -\tr(\bA_2\bM\bB_1\bM\bA_0\bM) & -\tr(\bA_1\bM\bB_2\bM\bA_0\bM) & -\tr(\bA_2\bM\bB_2\bM\bA_0\bM)\\
-\tr(\bA_2\bM\bB_1\bM\bA_0\bM) & -\tr(\bA_2\bM\bB_1\bM\bA_1\bM) & -\tr(\bA_2\bM\bB_2\bM\bA_0\bM) & -\tr(\bA_2\bM\bB_2\bM\bA_1\bM)\\
-\tr(\bA_1\bM\bB_2\bM\bA_0\bM) & -\tr(\bA_2\bM\bB_2\bM\bA_0\bM) & \tr(\bA_1\bM\bB_2\bM\bB_1\bM) & \tr(\bA_2\bM\bB_2\bM\bB_1\bM)\\
-\tr(\bA_2\bM\bB_2\bM\bA_0\bM) & -\tr(\bA_2\bM\bB_2\bM\bA_1\bM) & \tr(\bA_2\bM\bB_2\bM\bB_1\bM) & 0
\end{pmatrix}.
\end{align}
\end{small}
\subsubsection{Cubic-to-tetragonal transformations}
We now suppose that the transformation is cubic-to-tetragonal, with
\begin{align}\label{cubictetragonal}
\bU_1=\di(\eta_2,\eta_1,\eta_1), \quad \bU_2=\di(\eta_1,\eta_2,\eta_1), \quad \bU_3=\di(\eta_1,\eta_1,\eta_2),\quad \eta_1,\eta_2>0,\,\eta_1\neq\eta_2.
\end{align}
In order to prove that for general $\bm, \bR$ we cannot solve \eqref{gbcompat} we can without loss of generality assume that $i=1, k=2$ and that $\bU_1+\ba_1\otimes\bn_1=\bar\bQ\bU_2$ is one of the two twins connecting $\bU_1$ to $SO(3)\bU_2$. This is because for cubic-to tetragonal transformations every other twin  $\bU_i+\ba_i\otimes\bn_i=\hat\bQ\bU_k$ is such that
$$\bU_i=\tilde\bQ\bU_1\tilde\bQ^T,\,\ba_i=\tilde\bQ\ba_1,\,\bn_i=\tilde\bQ\bn_1, \,\bU_k=\tilde\bQ\bU_2\tilde\bQ^T$$
for some $\tilde\bQ\in P^{24}$, and therefore by multiplying \eqref{gbcompat} for $i=1$ on the left by $\tilde\bQ$ and on the right by $\tilde\bQ^T$ we have compatibility for $\bR$ replaced by  $\tilde\bQ\bR$ and $\bm$ replaced by $\tilde\bQ\bm$. In fact we will take below $\ba_1=\sqrt{2}\rho(-\eta_2,\eta_1,0)^T$, $\bn_1=\frac{1}{\sqrt{2}}(1,1,0)^T$, where $\rho:=\frac{\eta_2^2-\eta_1^2}{\eta_2^2+\eta_1^2}$.

On the other hand we need to take into account the different possibilities for the rank-one connections from $\bU_j$ to $SO(3)\bU_l$. There are six of these, two corresponding to $j=1$, $l=2$, two corresponding to $j=1$, $l=3$, and two corresponding to $j=2$, $l=3$. (The case $j=3$, $l=1$ is equivalent to $j=1$, $l=3$ since $\tilde\bQ\bU_1\tilde\bQ=\bU_3$ for the $180^\circ$ rotation $\tilde \bQ=-\bI+2\frac{\be_1-\be_3}{\sqrt{2}}\otimes\frac{\be_1-\be_3}{\sqrt{2}}$, and so in \eqref{gbcompat} we have that $\bQ(\bU_3+\ba_3\otimes\bn_3)\bR^T=\bQ\tilde\bQ(\bU_1+\tilde\bQ\ba_3\otimes\tilde\bQ\ba_3)\tilde\bQ\bR^T$ and we can replace $\bR$ by $\bR\tilde\bQ^T=\bR\tilde\bQ$. The cases $j=3$, $l=2$ and $j=2$, $l=1$ are similarly equivalent to $j=2$, $l=3$ and $j=1$, $l=2$ respectively.)
 
 The six rank-one connections are compound twins (i.e. both Type I and Type II) and are given in Table \ref{Table}. 
\begin{table}[hbt!]\label{Table}
$\begin{array}{c c c}
\hline
\text{Type} & \ba_j \:& \bn_j \\
\hline\\[-3mm]
\bU_1-\bU_2  &
\sqrt{2}\rho(-\eta_2, \,\pm\eta_1, \,0)^T &
\frac{1}{\sqrt{2}}(1,\pm 1,\,0)^T \\
\bU_1-\bU_3  &
\sqrt{2}\rho(-\eta_2, 0, \pm\eta_1)^T &
\frac{1}{\sqrt{2}}(1,\, 0,\pm 1)^T \\
\bU_2-\bU_3 &
\sqrt{2}\rho(0,\pm \eta_1,-\eta_2)^T &
\frac{1}{\sqrt{2}}(0,\pm 1, \,1)^T \\
\hline
\end{array}$\vspace{.05in}
\caption{Twins for cubic-to-tetragonal wells, where $\rho=\frac{\eta_2^2-\eta_1^2}{\eta_2^2+\eta_1^2}$.}
\end{table}
For each of the six possibilities we calculate  the determinant $\det\bcM_i$, $i=1,\dots,6,$ and take their product as follows
\begin{equation}
\label{Fprod}
\mathcal{F}(\bm,\bR)=\prod_{i=1}^6 \det\bcM_i(\bm,\bR).
\end{equation}
Note that for fixed $\eta_1,\eta_2$ each $\det\bcM_i(\bm,\bR)$ is a real analytic function defined on the 5-dimensional real analytic manifold $M:=S^2\times SO(3)$. Hence
 $\mathcal{F}(\bm,\bR)$ remains real analytic as a finite  product of real analytic functions.  

The following theorem shows that for general $(\bm,\bR)$ compatibility of zero-energy simple laminates across the grain boundary is impossible for any of the twinning systems. 
\begin{theorem}
\label{compatthm}
Let $\eta_1>0,\,\eta_2>0,\,\eta_1\neq\eta_2.$
The set of compatible interfaces for simple laminates
\[\mathcal{C}:=\{(\bm,\bR)\in S^2\times SO(3):\mathcal{F}=0\}\]
is closed and of measure zero (in particular nowhere dense). 
\end{theorem}
\begin{proof}
That $\mathcal{C}$ is closed follows immediately from the continuity of $F$. We recall that sets of measure zero are defined for any manifold, and that (see \cite[p125]{lee}) $\mathcal{C}$ is of measure zero if for some collection of real analytic charts $\{(U_\alpha, \varphi_\alpha)\}$ whose domains cover $M$ we have that $\varphi_\alpha(\mathcal{C}\cap U_\alpha)$ has measure zero in $\R^5$ for every $\alpha$. We apply the following result
\begin{proposition} \cite{mityagin} \label{Mityagin}
Let  $\mathcal{V}$ be a connected open subset of $\mathbb{R}^d$,  and $F:\mathcal{V}\to\R$ a real analytic function. If $F$ is not identically zero, then its zero set 
\[\mathcal{C}(F):=\{\bx\in\mathcal{V}:F(\bx)=0\}\]
has zero $d$-dimensional Lebesgue measure.
\end{proposition}
\noindent  to conclude that $\mathcal{C}$ has measure zero provided $\mathcal{F}:M\to\R$ is not identically zero, since the connectedness of $M$  and the identity theorem then implies that the real analytic function $F=\mathcal{F}\circ \varphi_\alpha^{-1}:\varphi_\alpha(U_\alpha)\to \R$ is not identically zero for every chart, and hence by Proposition \ref{Mityagin} $\varphi_\alpha(\mathcal{C}\cap U_\alpha)$ has measure zero.

To show that  $\mathcal{F}:M\to\R$ is not identically zero we may without loss of generality set $\eta_1=1$,  because the original problem of compatibility is homogeneous in the matrices $\bU_i$ and their entries. We perform a symbolic calculation in MATLAB to compute the product \eqref{Fprod}, which for $(\bm,\bR)\in S^2\times SO(3)$ is a polynomial in $\eta_2$.  Employing Lemma \ref{LemmaM}, the inputs are the matrix $\bM$ and the relative rotation $\bR$. Observe that for $\bm=(m_x,m_y,m_z)$ the matrix $\bM$ of Lemma \ref{LemmaM} is given by 
\begin{equation}\label{Mm}
\bM=
\begin{pmatrix}
0 & -m_z & m_y \\
m_z & 0 & -m_x\\
-m_y & m_x & 0
\end{pmatrix}.   
\end{equation}
For our computation we use $\bR=\bI-\sin\varphi\,\bK+(1-\cos\varphi)\,\bK^2,$ where 
\[\varphi=2\pi/3,\quad \mathbf{k} = (1/\sqrt{54},2/\sqrt{54},7/\sqrt{54})^T, \quad
\mathbf{K} = \begin{bmatrix} 
0 & -k_3 & k_2 \\ 
k_3 & 0 & -k_1 \\ 
-k_2 & k_1 & 0 
\end{bmatrix},\]
and $\bm=(1/\sqrt{14},2/\sqrt{14},3/\sqrt{14})^T.$
Each $\det\bcM_i$ is a polynomial of the form
\[\frac{a{\left({\eta_2}^2-1\right)}^{17}}{{\left(\eta_2^2+1\right)}^{\varkappa}}
P(\eta_2)\quad \text{where $\varkappa=8$ or $10,$ $P$ is a non-zero polynomial,  $a\neq 0$\: constant, }\]
and by taking the product we compute
\[\mathcal{F}(\bm,\bR)=\prod_{i=1}^6 \det\bcM_i
=\frac{a^{\prime}{\left(\eta_2^2-1\right)}^{102}}{{\left(\eta_2^2+1\right)}^{52}} 
\prod_{i=1}^6 P_i(\eta_2)\not\equiv 0,
\quad a^{\prime}\neq 0\:\:\text{constant,}\]
where $\ds\prod_{i=1}^6 P_i(\eta_2)$ is a non-zero polynomial of degree $28$. We repeat the calculation for a different set of inputs $(\bar\bM,\bar\bR)$, where 
$\bar\bR=\bI-\sin\varphi\,\bK+(1-\cos\varphi)\,\bK^2,$  
\[\varphi=\pi/4,\quad \mathbf{k} = (1/\sqrt{38},2/\sqrt{38},3/\sqrt{38})^T, \quad
\mathbf{K} = \begin{bmatrix} 
0 & k_3 & -k_2 \\ 
-k_3 & 0 & k_1 \\ 
k_2 & -k_1 & 0 
\end{bmatrix}\]
and $\bar\bM$ is given by \eqref{Mm} for $\bar\bm=(7/\sqrt{75},5/\sqrt{75},1/\sqrt{75})^T,$
to obtain a different polynomial 
\[\mathcal{F}(\bar\bm,\bar\bR)=\prod_{i=1}^6 \det\bar\bcM_i
=\frac{\bar a{\left(\eta_2^2-1\right)}^{102}}{{\left(\eta_2^2+1\right)}^{52}} 
\prod_{i=1}^6 \bar P_i(\eta_2)\not\equiv 0,
\quad \bar a\neq 0\:\:\text{constant.}\]
Here $\ds\prod_{i=1}^6 \bar P_i(\eta_2)$ is a non-zero polynomial of degree $28,$ so that both polynomials $\mathcal{F}(\bm,\bR)$ and $\mathcal{F}(\bar\bm,\bar\bR)$ are of the same degree. Next, to determine whether the polynomials
\begin{align}\label{Poly2}
P_1:=\prod_{i=1}^6 P_i(\eta_2) \quad\text{and}\quad P_2:=\prod_{i=1}^6 \bar P_i(\eta_2)
\end{align}
share a common root, we find the symbolic GCD (Greatest Common Divisor) polynomial between them. The GCD turns out to be a constant that does not depend on $\eta_2;$ for the set or parameters described above the GCD is $7.$ As a result, $P_1(\eta_2)$ and $P_2(\eta_2)$ do not share a common root, so that $\mathcal{F}$ is not identically zero for any $\eta_2$, completing the proof. 

We provide all the relevant code in Appendix \ref{Appendix}.
\end{proof}
\subsubsection{General transformations}
For general transformations with $N>3$ variants, such as cubic-to-orthorhombic or cubic-to-monoclinic, we could in principle proceed in the same way. Then the $\bU_i$ depend on more than two deformation parameters and we can form the finite product $\tilde{\mathcal{F}}$ of all the $\det\mathcal {M}_i(\bm,\bR)$ for all possible twin systems. A symbolic computation with different choices of $(\bm,\bR)$ would then presumably show that $\tilde{\mathcal{F}}$ is zero only  for algebraic relations between the deformation parameters which could not be simultaneously satisfied, so that again we would have that for any deformation parameters the set of $(\bm,\bR)$ allowing compatibility of simple laminates is closed and of measure zero.   

However without carrying out such an analysis we can obtain a weaker conclusion as follows. Consider, for example, a cubic-to-orthorhombic transformation with deformation parameters $\alpha>0,\beta>0,\gamma>0$ and corresponding variants 
\begin{equation}
\begin{aligned}
\label{cubicortho}
&\bU_1 =\left( \begin{array}{ccc} \alpha & 0 & 0 \\ 0 & \frac{\beta + \gamma}{2} &
\frac{\beta - \gamma}{2} \\ 0 &
\frac{\beta - \gamma}{2} & \frac{\beta + \gamma}{2}
\end{array} \right),
&\bU_2 =
\left( \begin{array}{ccc} \alpha & 0 & 0 \\ 0 & \frac{\beta + \gamma}{2} &
\frac{\gamma - \beta}{2} \\ 0 &
\frac{\gamma - \beta}{2} & \frac{\beta + \gamma}{2}
\end{array} \right),
 &\;\;\;\bU_3 =
\left( \begin{array}{ccc} \frac{\beta + \gamma}{2} & 0 & \frac{\beta -
\gamma}{2} \\ 0 & \alpha & 0 \\
\frac{\beta - \gamma}{2} & 0 & \frac{\beta + \gamma}{2}
\end{array} \right), \\ 
&\bU_4 =
\left( \begin{array}{ccc} \frac{\beta + \gamma}{2} & 0 & \frac{\gamma -
\beta}{2} \\ 0 & \alpha & 0 \\
\frac{\gamma - \beta}{2} & 0 & \frac{\beta + \gamma}{2}
\end{array} \right),
 &\bU_5 =\left( \begin{array}{ccc} \frac{\beta + \gamma}{2} & \frac{\beta -
\gamma}{2}
& 0 \\ \frac{\beta - \gamma}{2} &
\frac{\beta + \gamma}{2} & 0 \\ 0 & 0 & \alpha
\end{array} \right),
&\;\;\;\bU_6 =\left( \begin{array}{ccc} \frac{\beta + \gamma}{2} & \frac{\gamma -
\beta}{2}
& 0 \\ \frac{\gamma - \beta}{2} &
\frac{\beta + \gamma}{2} & 0 \\ 0 & 0 & \alpha
\end{array} \right).
\end{aligned}
\end{equation}
\begin{theorem}
\label{c2o}
For a cubic-to orthorhombic transformation, for almost all  deformation parameters $(\alpha,\beta,\gamma)\in(0,\infty)^3$ the set of $(\bm,\bR)$ for which compatibility of zero-energy simple laminates holds is closed and of measure zero.
\end{theorem}
\begin{proof}
$\tilde{\mathcal{F}}$  is a real analytic function of $(\alpha,\beta,\gamma,\bm,\bR)$ on $(0,\infty)^3\times M$, which is not identically zero since when $\beta=\gamma\neq\alpha$ we are in the cubic-to-tetragonal case studied above. Using Proposition \ref{Mityagin} we deduce that the set 
$$\tilde{\mathcal{C}}:=\{(\alpha,\beta,\gamma,\bm,\bR)\in (0,\infty)^3\times M:\tilde{\mathcal{F}}(\alpha,\beta,\gamma,\bm,\bR)=0\}$$ is of measure zero in the general sense of sets of measure zero for manifolds described at the beginning of the proof of Theorem \ref{compatthm}. Equivalently (see, for example, \cite[Proposition 16.5]{lee}), $\tilde{\mathcal{C}}$ has measure zero with respect to the product measure $\mathcal{L}^3\times\mathcal{H}^2\times\mu$ on $(0,\infty)^3\times M$, where $\mathcal{L}^3$ denotes 3-dimensional Lebesgue measure, $\mathcal{H}^2$ 2-dimensional Hausdorff measure on $S^2$ and $\mu$ Haar measure on $SO(3)$.   The result then follows from Fubini's Theorem.
\end{proof}

\subsection{Higher order laminates}\label{3.4}
For $m\geq 1$ let us define
\begin{equation}
\label{Emdef}
E_m=\{\bar\nu: \nu=\sum_{i=1}^m\lambda_i\delta_{\bX_i} \text{ a homogeneous GYM with }  \bX_i\in K, \lambda_i\geq 0, \sum_{i=1}^m\lambda_i=1\}.
\end{equation}
Note that (see e.g. \cite{j32})  by the weak continuity of minors 
$$E_2=\{\lambda\bX_1+(1-\lambda)\bX_2; \bX_1,\bX_2\in K, \rank(\bX_1-\bX_2)\leq 1, \lambda\in[0,1]\},$$
so that the problem studied in Section \ref{3.3} is whether for given $\bm,\bR$ we can solve
\begin{equation}
\label{E2E2}
\bA-\bB\bR=\bb\otimes\bm, 
\end{equation}
for $\bA,\bB\in E_2$ and $\bb\in\R^3$. As we saw, we cannot in general do this because we have nine equations for eight unknowns. Therefore it is natural to conjecture that by adding one additional volume fraction we will have the same number of unknowns and equations and be able to achieve compatibility. More generally, we conjecture that for generic $\bm,\bR$ we can solve \eqref{E2E2} for $\bA\in E_r,\bB\in E_s$ and $\bb\in\R^3$ if $r+s\geq 5$. This would imply, for example, that compatibility can be achieved with a simple laminate on one side of a planar grain boundary and a double laminate on the other side.

\section{The Taylor  set.}
\label{Taylor}
\subsection{Definition and properties}
\label{defprops}
Motivated by the definition given in \cite{bhattkohn96,bhattkohn97} for the geometrically linearized theory, and following \cite{m17}, we define the {\it Taylor  set}
$$\mS(K)=\bigcap_{\bR\in SO(3)}K^{\rm qc}\bR.$$
Note that since for $\bR\in SO(3)$ we have $(K\bR)^{qc}=K^{\rm qc}\bR$  (for more general statements in the same spirit see \cite{u5}), it follows that $\mS(K)$ is the intersection of quasiconvex sets and so is itself quasiconvex. Furthermore, since $\bQ K^{\rm qc}=(\bQ K)^{\rm qc}=K^{\rm qc}$ for all $\bQ\in SO(3)$, $\mS(K)$ is {\it isotropic}, that is $SO(3)\,\mS(K)\,SO(3)=\mS(K)$. Therefore
\begin{equation}
\label{deltadef}
\mathcal S(K)=SO(3)\Delta(K) SO(3), \text{ where  }
\Delta(K):=\{\bD\in\mS(K):\bD=\diag(v_1,v_2,v_3)>0\}.
\end{equation}
 The isotropy also implies that $\mS(K)$ is the zero set of a nonnegative isotropic quasiconvex function $\psi$; for example (see \cite{zhang1997,zhang98}) we can take $\psi$ to be the quasiconvexification of the isotropic  function $f(\bA):=\dist^2(\bA,\mS(K))$. 

Note that $K^{\rm qc}$ is not isotropic for $N>1$, since  there then exists $\bQ\in P^{24}$ with $\bQ\bU_1\bQ^T\neq \bU_1$, and thus $\bU_1\bQ\not\in K$. If $K^{\rm qc}$ were isotropic then we would have $\bU_1\bQ\in K^{\rm qc}$, and hence $\bU_1\bQ\in K^{\rm c}$. But then $\bU_1\bQ=\bar\nu$ for a probability measure $\nu$ with $\supp\nu\subset K$, and therefore $\int_K|\bU_1\bQ-\bF|^2d\nu(\bF)=2(|\bU_1|^2-\bU_1\bQ\cdot\bar\nu)=0$, so that $\bU_1\bQ\in K$, a contradiction.

Since by the minors relations $\det \bF=\det \bU_i$ for all $\bF\in K^{\rm qc}$ it follows that $\det\bF=\det\bU_i$ for all $\bF\in \mS(K)$.
Also, since we have assumed that the austenite has cubic symmetry, by the result of Bhattacharya \cite{bhattacharya92} $(\det \bU_i)^\frac{1}{3}\bI\in K^{\rm qc}$. Hence $(\det \bU_i)^\frac{1}{3}\bQ\in K^{\rm qc}$ for all $\bQ\in SO(3)$ and therefore $(\det \bU_i)^\frac{1}{3}SO(3)\subset\mS(K)$. In particular $\mS(K)$ is nonempty.  Furthermore, if $N>1$, that is if $\bU_i$ is not a dilatation,  $K^{\rm qc}$ always contains a nontrivial set of tetragonal wells \cite{bhattacharya92,j64}, and so by the result of Dolzmann \& Kirchheim \cite{DoKir} considered in more detail below, an open neighbourhood of $(\det \bU_i)^\frac{1}{3}SO(3)$ in $\{\bA\in \R^{3\times 3}:\det\bA=\det\bU_i\}\subset\mS(K)$.

The following easy result gives an upper bound for $\mS(K)$.
\begin{proposition}
\label{upperbound}
$\mS(K)\subset SO(3){\mathcal D(K)}SO(3),$
where $$\mathcal D(K):=\{\bD\in K^{\rm qc}: \bD=\diag(v_1,v_2,v_3)>0\}.$$
\end{proposition}
\begin{proof}
Let $\bA\in \mS(K)$. Then $\bA=\bQ\bD\bR$ for some $\bQ,\bR\in SO(3)$ and diagonal $\bD>0$. Since $\mS(K)\subset K^{\rm qc}\bR$, we have that $\bQ\bD\in K^{\rm qc}$ and hence $\bD\in K^{\rm qc}$, giving the result.
\end{proof}

The importance of the Taylor set lies in the fact that macroscopic deformations $\by$ with $D\by(\bx)\in \mS(K)$ a.e. correspond to zero energy microstructures  for arbitrary grain geometry and grain rotations.
\begin{theorem}
\label{Sprop}
Let $\by\in W^{1,\infty}(\Omega,\R^3)$ satisfy $D\by(\bx)\in\mS(K)$ for a.e. $\bx\in\Omega$. Then there exists a bounded sequence $\by^{(j)}\in W^{1,\infty}(\Omega,\R^3)$ such that $D\by^{(j)}$ generates a GYM $(\nu_{\bx})_{\bx\in\Omega}$ with $\supp\nu_{\bx}\subset K$  and $\bar\nu_{\bx}=D\by(\bx)$ for a.e. $\bx\in\om$.
\end{theorem}
\begin{proof}
By \cite[Lemma 4.13]{rindler2018}, for each grain $\om_i$ there is a bounded sequence $\by^{(j)}_i\in W^{1,\infty}(\om_i,\R^3)$  satisfying the boundary condition $\left.\by^{(j)}_i\right|_{\partial\om_i}=\by$ that generates the gradient Young measure $(\nu_{\bx})_{\bx\in\om_i}$ with $\bar\nu_{\bx}=D\by(\bx)$ for a.e. $\bx\in\om_i$. Then the sequence $\by^{(j)}(\bx):=\by_i^{(j)}(\bx)$ for $\bx\in\om_i$ has the required properties.
\end{proof}
\begin{remark} \rm A more satisfactory statement of Theorem \ref{Sprop} would be that if in addition $\by$ is invertible, then the $\by^{(j)}$ can also be chosen to be invertible, but it is not clear how to prove this.
\end{remark}
\subsection{The two-well problem.}
The only case when $K^{\rm qc}$ is known is for two energy wells. It is well known that this cannot occur for cubic austenite with transformation strain $\bU$ (because it would imply that the subgroup of $P^{24}$ consisting of the $\bQ\in P^{24}$ such that $\bQ^T\bU\bQ=\bU$ has order 12, and so is the tetragonal group consisting of the identity, $180^\circ$ rotations about the cubic axes and $120^\circ$ rotations about the cube diagonals, from which it follows that $\bU$ is a multiple of the identity). However it can occur for tetragonal-to-orthorhombic transformations or special orthorhombic-to-monoclinic transformations.
 
In the two-well problem we can take
\begin{equation}
\label{K2}
K=SO(3)\bU_1\cup SO(3)\bU_2,
\end{equation}
where 
\begin{align*}
\bU_1=\di(\eta_1,\eta_2,\eta_3), \quad \bU_2=\di(\eta_2,\eta_1,\eta_3), \quad \text{and} \quad  \eta_2>\eta_1>0,\:\:\eta_3>0
\end{align*}
and $K^{\rm qc}$ consists of matrices $F\in GL^+(3)$ such that
\begin{align*}
\bF^T\bF=\begin{pmatrix}
a & c & 0\\
c &  b & 0\\
0 & 0 & \eta_3^2
\end{pmatrix},
\quad \text{where} \:\:\: a>0,\:\:\: b>0,\:\:\: a+b+2|c|\leq\eta_1^2+\eta_2^2,\:\:\: ab-c^2=\eta_1^2\eta_2^2.
\end{align*}
Observe that for $c=0$ the above becomes the set of diagonal matrices such that
\begin{align}
\label{BALL-JAMES}
\bF^T\bF=\begin{pmatrix}
a & 0 & 0\\
0 &  b & 0\\
0 & 0 & \eta_3^2
\end{pmatrix},
\quad \text{where} \:\:\: a>0,\:\:\: b>0,\:\:\: a+b\leq\eta_1^2+\eta_2^2,\:\:\: ab=\eta_1^2\eta_2^2.
\end{align}
\begin{theorem}
\label{2DSK}
Let $K$ be as in \eqref{K2}. Then
\begin{align}
\mathcal{S}(K)=\begin{cases}
            \emptyset    & \quad \text{if }\: \eta_3\neq\sqrt{\eta_1\eta_2},\\
          \eta_3  SO(3)  & \quad \text{if }\: \eta_3=\sqrt{\eta_1\eta_2}.
            \end{cases}    
\end{align}
\end{theorem}
\begin{proof}
Let $\bF\in\mS(K)$.  Then $\bF=\bQ\bD\bQ'$ for $\bQ,\bQ'\in SO(3)$, and $\bD=\di(d_1,d_2,d_3)>0$. Since $SO(3)\,\mS(K)\, SO(3)=\mS(K)$ we have that $\bD\in\mS(K)$, and thus $\bR\bD\bR^T\in\mS(K)\subset K^{\rm qc}$  for any $\bR\in SO(3)$, so that  there exist $a>0,b>0,c$ (depending on $\bR$) with $ab-c^2=\eta_1^2\eta_2^2,$ $a+b+|2c|\leq\eta_1^2+\eta_2^2$ and
\begin{equation}
\label{abc}
\begin{pmatrix}
a & c & 0\\
c &  b & 0\\
0 & 0 & \eta_3^2
\end{pmatrix}
=\bR
\begin{pmatrix}
d_1^2 & 0 & 0\\
0 &  d_2^2 & 0\\
0 & 0 & d_3^2
\end{pmatrix}
\bR^T.
\end{equation}
Choosing $\bR=\bI$ we get that $d_3=\eta_3$. Choosing $$\bR=\left(\begin{array}{ccc}-1&0&0\\0&0&1\\0&1&0\end{array}\right)$$ we get that $d_2=\eta_3$, and similarly $d_1=\eta_3$. 
Hence both sides of \eqref{abc} equal $\eta_3^2\bI$ so that we must have $a=b=\eta_3^2,$ $c=0$. Thus $\eta_3=\sqrt{\eta_1\eta_2}$ when since $2\eta_3^2+0\leq\eta_1^2+\eta_2^2$ we obtain $\eta_3\bI\in K^{\rm qc}$ and $\mS(K)=\eta_3SO(3)$.
\end{proof}
\begin{remark} \rm As proved in Theorem \ref{Sprop}, deformations with $D\by(\bx)\in\mS(K)$ a.e. can be realized by underlying zero energy microstructures. However 
 there could be additional zero energy microstructures corresponding to different macroscopic gradients for particular grain geometries and rotations.
\end{remark}
Now we consider the set
\begin{align*}
\mathcal{S}_{2D}(K)=\bigcap_{\substack{\bR\in SO_{2D}(3)}} K^{\rm qc}\bR,
\end{align*}
where $SO_{2D}(3):=\{\bR\in SO(3): \bR\be_3=\be_3\}$, 
corresponding to a 2D situation in which the grain rotations $\bR$ all have axis $\be_3$.
The following result is closely related to work of Kohn \& Niethammer \cite{kohnniethammer00} and Dolzmann \cite{dolzmann03} (we note that while  the relevant Theorem 2.4.1 in \cite{dolzmann03} is clearly motivated and the proof  conceptually correct the expression for the set corresponding to $\mathcal{S}_{2D}(K)$ contains an algebraic error).
\begin{theorem}
\label{2DSK1}
\begin{align*}
\mathcal{S}_{2D}(K)
&=\Bigl\{\bF=\bR\bD\tilde\bR: \bR\in SO(3),\tilde\bR\in SO_{2D}(3), \bD=\diag(v_1,v_2,\eta_3), \\
& \hspace{2.5in}v_1v_2=\eta_1\eta_2,\; 0<v_i\leq\sqrt{\frac{\eta_1^2+\eta_2^2}{2}}, i=1,2\Bigr\}
\end{align*}
\end{theorem}
\begin{proof}
Rotations $\bR\in SO_{2D}(3)$ have the form
\begin{equation*}
\label{e3axis}
\bR=\left(\begin{array}{ccc}\cos\theta&-\sin\theta&0\\\sin\theta&\cos\theta&0\\0&0&1\end{array}\right), \text{   where }\theta\in[0,2\pi].
\end{equation*}
Thus $\bF\in\mS_{2D}(K)$ iff $\bF\bR^T\in K^{\rm qc}$ for all such $\bR$,  that is
\begin{align}
\label{strainform}\bF^T\bF&=\left(\begin{array}{ccc}\cos\theta&\sin\theta&0\\-\sin\theta&\cos\theta&0\\0&0&1\end{array}\right)\left(\begin{array}{ccc}a&c&0\\c&b&0\\0&0&\eta_3^2\end{array}\right)\left(\begin{array}{ccc}\cos\theta&-\sin\theta&0\\\sin\theta&\cos\theta&0\\0&0&1\end{array}\right),
\end{align}
where $a,b,c$ (depending on $\theta$) satisfy $a>0, b>0, a+b+|2c|\leq\eta_1^2+\eta_2^2, ab-c^2=\eta_1^2\eta_2^2$. Hence
$$\bF^T\bF=\left(\begin{array}{ccc}U_{11}&U_{12}&0\\U_{21}&U_{22}&0\\0&0&\eta_3^2\end{array}\right),$$
where $\bU=(U_{ij})\in GL^+(2)$ is symmetric and positive definite, and can thus be diagonalized so that $\bQ\bU\bQ^T=\diag(v_1,v_2)$ for some $\bQ\in SO(2)$, where $v_1>0, v_2>0$. Hence $$\mS_{2D}(K)=\{SO(3)\bD SO_{2D}(3): \bD =\diag(v_1,v_2,\eta_3)\in \mS_{2D}(K)\}.$$
From \eqref{strainform} we deduce that $\bD=\diag(v_1,v_2,\eta_3)\in \mS_{2D}$ if and only if for each $\theta\in[0,2\pi]$
$$\left(\begin{array}{cc}\cos^2\theta\, v_1^2+\sin^2\theta\, v_2^2&\cos\theta\sin\theta\,(v_1^2-v_2^2)\\\cos\theta\sin\theta\,(v_1^2-v_2^2)&\sin^2\theta \,v_1^2+\cos^2\theta\, v_2^2\end{array}\right)=\left(\begin{array}{cc}a&c\\c&b\end{array}\right)$$
for some $a,b$  satisfying the above conditions, that is if and only if $v_1v_2=\eta_1\eta_2$ and $$v_1^2+v_2^2+|\sin 2\theta\,(v_1^2-v_2^2)|\leq \eta_1^2+\eta_2^2\text{  for all }\theta,$$ giving the result.
\end{proof}
\begin{remarks}\rm\mbox{ }\newline$(i)$ 
If we define $SO^-_{2D}(3)=\{\bR\in SO(3):\bR\be_3=-\be_3\}$ then it is easily checked that $SO^-_{2D}(3)=\bR_0SO_{2D}(3)$, where $\bR_0:=\diag(1,-1,-1)$. Since if $\tilde\bR\in SO^-_{2D}(3)$ and $\bR\in SO(3)$ we have that $\bR\bD\tilde\bR=(\bR\bR_0)\bD(\bR_0\tilde\bR)$ it follows that also
$$\mS_{2D}(K)=\bigcap_{\bR\in SO^-_{2D}(3)}K^{\rm qc}\bR.$$
$(ii)$
There are nontrivial deformations $\by$ with $D\by(\bx)\in \mathcal{S}_{2D}(K)$ a.e. $\bx=(x_1,x_2,x_3)\in\Omega$ such as
\[\by(\bx)=(\sqrt{\eta_1\eta_2}x_1,\sqrt{\eta_1\eta_2}x_2,\eta_3x_3)+\varepsilon g(\bx\cdot\be^{\perp})\be,\]
where $\{\be,\be^{\perp},\be_3\}$ is an orthonormal basis of $\R^3$, $g:\R\to\R$ is Lipschitz with $|g^{\prime}|\leq M<\infty$ and $|\varepsilon|$ is sufficiently small. Such deformations deform non-trivially the grain boundaries. 
\end{remarks}
\subsection {Cubic-to-tetragonal transformations}
\label{cubictetragonala}
For a cubic-to-tetragonal transformation with $\bU_1$, $\bU_2$, $\bU_3$ given by \eqref{cubictetragonal}
 we set $\eta_{\rm min}:=\min(\eta_1,\eta_2)$, $\eta_{\rm max}=\max(\eta_1,\eta_2)$.
In this case the quasiconvexification $K^{\rm qc}$ of $K=\bigcup^3_{i=1}SO(3)\bU_i$ is not known.

A  lower bound for $\mS(K)$ follows from the result of Dolzmann \& Kirchheim \cite{DoKir} that if $\frac{\eta_2}{\eta_1}<\frac{27}{8}$ then
\begin{align} \label{lowerbound}
B((\eta_1^2\eta_2)^{\frac{1}{3}}\bI,\rho)\cap\{\bA\in M^{3\times 3}:\det \bA=\eta_1^2\eta_2\}& \subset K^{\rm qc} \text{ for }
\rho:=\frac{\eta_2^{\frac{1}{3}}(\eta_2^{\frac{1}{3}}-\eta_1^{\frac{1}{3}})^2}{62}.
\end{align}
(See also \cite{ContiDolzmannKirchheim2007}, which gives a more general result but without a specific value for $\rho$.)
Thus, since $SO(3)B((\eta_1^2\eta_2)^{\frac{1}{3}}\bI,\rho)=B((\eta_1^2\eta_2)^{\frac{1}{3}}\bI,\rho)SO(3)$ and $SO(3)K=K$, we obtain that
\begin{equation}
\label{lowerbound1}
SO(3)\left(B((\eta_1^2\eta_2)^{\frac{1}{3}}\bI,\rho)\cap\{\bA\in M^{3\times 3}:\det \bA=\eta_1^2\eta_2\}\right)\subset \mS(K).
\end{equation}

 The upper bound in Proposition \ref{upperbound} can be computed explicitly on account of the following result of Peigney that can be found, but is not stated as a theorem, in \cite[p1500]{Peigney2013}, for which we give a simple proof. In Theorem \ref{better} we will give a better upper bound, but we include the result and its proof for their independent interest.
\begin{theorem}
\label{peigney}
Let $\ds K=\bigcup^3_{i=1}SO(3)\bU_i$, where the $\bU_i$ are given in \eqref{cubictetragonal}. Then 
\begin{align}
\label{DELTA}
\mathcal D(K)=\{\bD=\diag(v_1,v_2,v_3): v_1v_2v_3=\eta_1^2\eta_2,\:\eta_{\rm min}\leq v_i\leq\eta_{\rm max}\}.
\end{align}
\end{theorem}
\begin{proof}
Let $\mathcal D(K)':=\{\bD=\diag(v_1,v_2,v_3): v_1v_2v_3=\eta_1^2\eta_2,\:\eta_{\rm min}\leq v_i\leq\eta_{\rm max}\}$. If $\bD=\diag(v_1,v_2,v_3)>0$ belongs to $K^{\rm qc}$ then $\bD=\bar\nu$ for a homogeneous GYM $\nu$ supported in $K$. By the minors relations we have that $$\det\bD=\det \bar \nu=\langle\nu,\det\rangle=\int_K\det \bA\, d\nu(\bA)=\eta_1^2\eta_2.$$
Next we
follow the method of \cite{j48} and note that the largest singular value $v_{\rm max}(\bA)$ of $\bA\in\R^{3\times 3}$ is a convex function of $\bA$, so that by Jensen's inequality
$$\max_iv_i=v_{\rm max}(\bD)=v_{\rm max}\left(\int_K\bA\,d\nu(\bA)\right)\leq \int_Kv_{\rm max}(\bA)\,d\nu(\bA) \leq\max_{\bA\in K}v_{\rm max}(\bA)=\eta_{\rm max}.$$
Using the minors relation $\cof \bar\nu=\langle \nu,\cof\rangle$ we similarly obtain
$$\max_i\frac{v_1v_2v_3}{v_i}=v_{\rm max}(\cof\bD)=v_{\rm max}\left(\int_K\cof \bA\,d\nu(\bA)\right)\leq\int_K v_{\rm max}(\cof\bA)\,d\nu(\bA)\leq\frac{\eta_1^2\eta_2}{\eta_{\rm min}},$$
so that $\min_iv_i\geq \eta_{\rm min}$.
Therefore $\mathcal D(K)\subset \mathcal D(K)'$.

Now let $\bD\in \mathcal D(K)'$. Considering the two wells $SO(3)\bU_1,\,SO(3)\bU_2$ and applying \eqref{BALL-JAMES} we find that $\bD^{(1)}:=\diag(v_1,\frac{\eta_1\eta_2}{v_1},\eta_1)\in K^{\rm qc}$, because the inequality $v_1^2+\left(\frac{\eta_1\eta_2}{v_1}\right)^2\leq \eta_1^2+\eta_2^2$ can be rewritten as $(v_1^2-\eta_1^2)(v_1^2-\eta_2^2)\leq 0$, which holds because $\eta_{\rm min}\leq v_1\leq\eta_{\rm max}$. Similarly, considering the two wells  $SO(3)\bU_1,\,SO(3)\bU_3$ we find that  $\bD^{(2)}:=\diag(v_1,\eta_1,\frac{\eta_1\eta_2}{v_1})\in K^{\rm qc}$. We now apply \eqref{BALL-JAMES} to the two wells $SO(3)\bD^{(1)},\,SO(3)\bD^{(2)}$ to deduce that $\bD=\diag(v_1,v_2,v_3)\in K^{\rm qc}$, because the inequality $v_2^2+v_3^2\leq \eta_1^2+(\frac{\eta_1\eta_2}{v_1})^2$ in conjunction with $v_1v_2v_3=\eta_1^2\eta_2$ takes the form $(v_2^2-\eta_1^2)(v_3^2-\eta_1^2)\geq 0$, which holds because $\eta_{\rm min}\leq v_i\leq\eta_{\rm max}$. Thus $\mathcal D(K)'\subset\mathcal D(K)$, completing the proof.
\end{proof}
From Proposition \ref{upperbound} and Theorem \ref{peigney} we see that an upper bound for $\mS(K)$ is
$SO(3)\mathcal D(K) SO(3),$ where $\mathcal D(K)$ is the set defined in \eqref{DELTA}. Equivalently $\Delta(K)\subset\mathcal D(K)$. We now establish a better upper bound.
\begin{theorem}
\label{better}
If $\bD=\diag(v_1,v_2,v_3)\in \Delta(K)$ then $v_1v_2v_3=\eta_1^2\eta_2$,
\begin{equation}
\label{bounds}
\frac{\sqrt 3\,\eta_1\eta_2}{\sqrt{\eta_1^2+2\eta_2^2}}\leq v_i\leq\frac{\sqrt{\eta_2^2+2\eta_1^2}}{\sqrt 3},
\end{equation}
and
\begin{align}\label{sumvi}
v_1+v_2+v_3&\leq\eta_2+ 2\eta_1\\\label{sumvivj}
v_1v_2+v_2v_3+v_3v_1&\leq \eta_1(\eta_1+2\eta_2).
\end{align}
\end{theorem}
\begin{proof}
Let $\bD=\diag(v_1,v_2,v_3)\in \Delta(K)$. Since $\bD\in K^{\rm qc}$, by the minors relations $v_1v_2v_3=\det\bD=\eta_1^2\eta_2$. Also we have that  $\bD\in\bigcap_{\bR\in SO(3)}K^{\rm c}\bR$. Given $i$ choose $\bR\in SO(3)$ with $\bR\be_i=\be:=\frac{1}{\sqrt 3}(\be_1+\be_2+\be_3)$. Then $\bD\bR^T\in K^{\rm c}$ and so  $\bD\bR^T=\bar\nu_{\bR}=\int_K\bF d\nu_{\bR}(\bF)$ for a probability measure $\nu_{\bR}$ with $\supp\nu_{\bR}\subset K$. Therefore
$$v_i=\bD\bR^T\be\cdot\be_i=\int_K\bF\be\cdot\be_i\, d\nu_{\bR}(\bF)\leq \int_K|\bF\be|d\nu_{\bR}(\bF)=\frac{\sqrt{\eta_2^2+2\eta_1^2}}{\sqrt 3},$$
since if $\bF\in K$ then $\bF=\bQ\bU_j$ for some $\bQ\in SO(3)$ and $j$, so that $|\bF\be|^2=|\bU_j\be|^2=\frac{1}{3}(\eta_2^2+2\eta_1^2)$. This proves the upper bound in \eqref{bounds}.

For $\bF=\bQ\bU_j\in K$ we have that $\tr\bF\leq\tr\bU_j=\eta_2+2\eta_1$. Since $\bD=\int_K\bF\,d\nu_{\bI}(\bF)$ for a probability measure $\nu_{\bI}$ with $\supp \nu_{\bI}\subset K$ it follows that $\tr \bD\leq\int_K\tr\bF\,d\nu_{\bI}(\bF)\leq \eta_2+2\eta_1$, giving \eqref{sumvi}.

To obtain the lower bound in \eqref{bounds} and the estimate \eqref{sumvivj} we note that $\bD\in \Delta(K)$ implies that, for any $\bR\in SO(3)$, $\bD\bR^T=\bar\nu_{\bR}$ for a homogeneous GYM $\nu_{\bR}$ with $\supp \nu_{\bR}\subset K$. Therefore, by the minors relations
$$(\cof \bD)\bR^T=\cof (\bD\bR^T)=\cof\bar\nu_{\bR}=\langle\nu_{\bR},\cof\rangle\in (\cof K)^{\rm c},$$
and hence $\cof\bD\in\bigcap_{\bR\in SO(3)}(\cof K)^{\rm c}\bR$. Since $$\cof K=SO(3)\diag(\eta_1^2,\eta_1\eta_2,\eta_1\eta_2)\cup SO(3)\diag(\eta_1\eta_2,\eta_1^2,\eta_1\eta_2)\cup SO(3)\diag(\eta_1\eta_2,\eta_1\eta_2,\eta_1^2)$$
we can replace $v_1,v_2,v_3$ by $v_2v_3,v_1v_3,v_1v_2$ and $\eta_1,\eta_2$ by $\eta_1\eta_2, \eta_1^2$ respectively in the estimates obtained above. Thus from the upper bound in \eqref{bounds}, and using $v_1v_2v_3=\eta^2_1\eta_2$, we have that 
$$\frac{\eta_1^2\eta_2}{v_i}\leq \frac{\sqrt{\eta_1^4+2\eta_1^2\eta_2^2}}{\sqrt 3},$$
giving the lower bound in \eqref{bounds}. Similarly, from \eqref{sumvi} we obtain \eqref{sumvivj}.
\end{proof}
 
\begin{remarks}\rm
\label{Sradius}\mbox{ }\newline
$(i)$ Since the proof uses only convexity and the minors relations, the bounds \eqref{better}, \eqref{sumvi} hold for  $\bD$ in the possibly larger set $\displaystyle\bigcap_{\bR\in SO(3)}K^{\rm pc}\bR$, where $$K^{\rm pc}:=\{\bA\in \R^{3\times 3}:\varphi(\bA)\leq\max_K\varphi \text{ for all polyconvex }\varphi\}.$$\\
$(ii)$ The bounds \eqref{bounds} imply that $\eta_{\rm min}\leq v_i\leq\eta_{\rm max}$, the upper bound being obvious and the lower bound following by writing $\frac{\sqrt{3}\eta_1\eta_2}{\sqrt{\eta_1^2+2\eta_2^2}}=\frac{\sqrt{3}}{\sqrt{\eta_2^{-2}+2\eta_1^{-2}}}$.\\
$(iii)$ From \eqref{bounds}, \eqref{sumvi} and $v_1^2+v_2^2+v_3^2\leq (\max_iv_i)(v_1+v_2+v_3)$ we obtain
\begin{equation}
v_1^2+v_2^2+v_3^2\leq\frac{\sqrt{\eta_2^2+2\eta_1^2}}{\sqrt {3}}\left(\eta_2+2\eta_1\right),
\end{equation}
from which the estimate
\begin{equation}
\label{distfromsphere}
\eta_2^2+2\eta_1^2-|\bD|^2\geq\frac{2}{3+\sqrt{3}\,h(\eta_1,\eta_2)}|\eta_1-\eta_2|^2,\text{ where } h(\eta_1,\eta_2):=\frac{\eta_2+2\eta_1}{\sqrt{\eta_2^2+2\eta_1^2}},
\end{equation}
follows, giving a positive lower bound on the distance of $\mS(K)$ from the sphere of radius $\sqrt{\eta_2^2 +2\eta_1^2}$ in $\R^{3\times 3}$ on which $K$ lies.\\
$(iv)$ From \eqref{bounds} we can estimate the diameter of $\Delta(K)$. Suppose $\bD=\diag(v_1,v_2,v_3)$ and $\bD'=\diag(v_1',v_2',v_3')$ belong to $\Delta (K)$. Then
$$|\bD-\bD'|^2=\sum_{i=1}^3(v_i-v_i')^2,$$
so that by \eqref{bounds}
\begin{equation}\label{diam} \diam \Delta(K)\leq g(\eta_1,\eta_2), \text{ where }g(\eta_1,\eta_2):=\sqrt{\eta_2^2+2\eta_1^2}-\frac{3\eta_1\eta_2}{\sqrt{\eta_1^2+2\eta_2^2}}.
\end{equation}
Note that $g$ is a smooth function with $g(1,1)=0$ and $\nabla g(1,1)=0$. Hence for $(\eta_1,\eta_2)$ in a neighbourhood of $(1,1)$ we have that 
\begin{equation}
\label{diambound}
\diam \Delta(K)\leq C\left((\eta_1-1)^2+(\eta_2-1)^2\right)
\end{equation}
  for some constant $C>0$. Note that this is consistent with the lower bound \eqref{lowerbound1} since $\rho\leq C'(\eta_1-\eta_2)^2$ in a neighbourhood of $(1,1)$ for some constant $C'$. 

 In the linearized theory we assume the deformation to have the form $\by(\bx)=\bx+\ep \bu$, where $\ep>0$ is a small parameter, with the energy wells  scaled similarly by $\bU_i=\bI+\ep \bE_i$, $1\leq i\leq N$, for symmetric linear strains $\bE_i$  having equal trace (see, for example, \cite{j40,bhattacharya93,kohn91,schmidt08}). Then, setting $K_{\ep}:=\bigcup_{i=1}^NSO(3)(\bI+\ep\bE_i)$ (see \cite{bhattkohn97}) $K_{\ep}^{\rm qc}$ is replaced by the convex hull $K_{\rm lin}^{\rm c}$ of $K_{\rm lin}:=\{\bE_1,\ldots,\bE_N\}$, while $\Delta(K_{\ep})$ is replaced by $\Delta(K_{\rm lin}):=\{\bD:\bD=\diag(v_1,v_2,v_3)>0, \bR\bD\bR^T\in   K_{\rm lin}^{\rm c} \text{ for all }\bR\in SO(3)\}$. In the cubic-to-tetragonal case  \eqref{diambound} implies that $\lim_{\ep\to 0} \frac{1}{\ep}\diam(\Delta(K_{\ep}))=0$, so that \eqref{diam} can be regarded as a nonlinear version of the result of Bhattacharya \& Kohn \cite{bhattkohn97} that $\Delta(K_{\rm lin})$ consists of the single matrix $\frac{1}{3}(\tr\bE_i)\bI$.
 
 \end{remarks}

\subsection{Cubic-to-orthorhombic transformations}
\label{Rem9}
For cubic-to-orthorhombic transformations, with variants $\bU_1,\ldots,\bU_6$ given by \eqref{cubicortho} and deformation parameters $\alpha>0,\beta>0,\gamma>0$, we do not know an equivalent result to Theorem \ref{peigney}  characterizing the positive diagonal matrices $\bD=\diag(v_1,v_2,v_3)$ belonging to $K^{\rm qc}$. Using the maximum singular value as in the proof of Theorem \ref{peigney} we have the necessary condition that  $\bD\in K^{\rm qc}$ implies
\begin{equation}
\label{orbound}
\min(\alpha,\beta,\gamma)\leq v_i\leq\max(\alpha,\beta,\gamma).
\end{equation}
Also, noting that $\bU_1=\bQ_{\pi/4}\diag (\alpha,\beta,\gamma)\bQ_{\pi/4}^T$ and $\bU_2=\bQ_{\pi/4}\diag (\alpha,\gamma,\beta)\bQ_{\pi/4}^T$, where $$\bQ_{\pi/4}:=\left(\begin{array}{ccc}1&0&0\\0&\frac{1}{\sqrt 2}&-\frac{1}{\sqrt 2}\\0&\frac{1}{\sqrt 2}&\frac{1}{\sqrt 2}\end{array}\right),$$
and applying \eqref{BALL-JAMES} we have that $\bQ_{\pi/4}\diag(\alpha,\sqrt{\beta\gamma},\sqrt{\beta\gamma})\bQ_{\pi/4}^T=\diag(\alpha,\sqrt{\beta\gamma},\sqrt{\beta\gamma})\in K^{\rm qc}$. Arguing similarly for the other $\bU_i$ we see that 
\begin{equation}
\label{tetrwells}
SO(3)\diag(\alpha,\sqrt{\beta\gamma},\sqrt{\beta\gamma})\cup SO(3)\diag(\sqrt{\beta\gamma},\alpha,\sqrt{\beta\gamma})\cup SO(3)\diag(\sqrt{\beta\gamma},\sqrt{\beta\gamma},\alpha)\subset K^{\rm qc}.\end{equation}
Hence, by Theorem \ref{peigney} we obtain the sufficient condition that if 
$$\min(\sqrt{\beta\gamma},\alpha)\leq v_i\leq\max(\sqrt{\beta\gamma},\alpha)$$
then $\bD\in K^{\rm qc}$.

From \eqref{tetrwells} and \eqref{lowerbound} we deduce the lower bound that if $\frac{\alpha}{\sqrt{\beta\gamma}}<\frac{27}{8}$ then
\begin{align} \label{lowerboundor}
B((\alpha\beta\gamma)^{\frac{1}{3}}\bI,\bar\rho)\cap\{\bA\in M^{3\times 3}:\det \bA=\alpha\beta\gamma\}& \subset K^{\rm qc} \text{ for }
\bar\rho:=\frac{\alpha^{\frac{1}{3}}(\alpha^{\frac{1}{3}}-(\beta\gamma)^{\frac{1}{6}})^2}{62}.
\end{align}
(More generally, for cubic austenite and any transformation strain $\bU$ conditions can be deduced from  \cite[Lemma 3.5]{j64}  under which $K^{\rm qc}$ contains a nontrivial set of cubic-to-tetragonal wells, so that $\mS(K)$ contains a relatively open neighbourhood of $\{\bA\in M^{3\times 3}:\det \bA=\det\bU\}$.)

 As regards an upper bound for $\mS(K)$ the following result is the analogue of Theorem \ref{better} (in fact a generalization of it since we can take $\beta=\gamma=\eta_1, \alpha=\eta_2$).
\begin{theorem}
\label{betteror}
If $\bD=\diag(v_1,v_2,v_3)\in \Delta(K)$ then $v_1v_2v_3=\alpha\beta\gamma$,
\begin{equation}
\label{boundsor}
\frac{\sqrt 3\,\alpha\beta\gamma}{\sqrt{2\alpha^2\max(\beta,\gamma)^2+\beta^2\gamma^2}}\leq v_i\leq\frac{\sqrt{2\max(\beta,\gamma)^2+\alpha^2}}{\sqrt 3},
\end{equation}
and
\begin{align}\label{sumvior}
v_1+v_2+v_3&\leq \alpha+\beta+\gamma\\\label{sumvivjor}
v_1v_2+v_2v_3+v_3v_1&\leq\beta\gamma+\alpha\beta+\gamma\alpha.
\end{align}
\end{theorem}
\begin{proof}
This follows exactly the same method as that of Theorem \ref{better}. The upper bound in \eqref{boundsor} is obtained using the same rotation $\bR$, while \eqref{sumvior} follows using $\tr \bF\leq\tr \bU_i$ for $\bF\in K$. Then the lower bound in \eqref{boundsor} and \eqref{sumvivjor} follow by replacing $v_1,v_2,v_3$ by $v_2v_3,v_3v_1, v_1v_2$ and $\alpha, \beta,\gamma$ by $\alpha\gamma, \gamma\beta, \alpha\beta$ respectively, where we have used $\cof\bD\in\bigcap_{\bR\in SO(3)}(\cof K)^{\rm c}\bR$ and that, for example,
$$\cof \bU_1=\left(\begin{array}{ccc}\gamma\beta&0&0\\0&\frac{\alpha(\beta+\gamma)}{2}&\frac{\alpha(\gamma-\beta)}{2}\\0&\frac{\alpha(\gamma-\beta)}{2}&\frac{\alpha(\beta+\gamma)}{2}\end{array}\right).$$
\end{proof}
The same results hold for the other cubic-to-orthorhombic transformation \eqref{lesscommon}, with \eqref{boundsor} replaced by the stronger inequalities
\begin{equation}
\label{boundsor1}
\frac{\sqrt 3\,\alpha\beta\gamma}{\sqrt{\alpha^2\beta^2+\beta^2\gamma^2+\gamma^2\alpha^2}}\leq v_i\leq\frac{\sqrt{\alpha^2+\beta^2+\gamma^2}}{\sqrt 3}.
\end{equation}

\section{Discussion}
\label{disc}
The results of Section \ref{3} explain why complex microstructures are necessary in polycrystals  to achieve compatibility at grain boundaries, and in particular why simple laminates are not sufficient, consistent with the observation of doubly laminated microstructures such as for the cubic-to-tetragonal transformations in ${\rm BaTiO_3}$ (see \cite{Arlt90}) and in  ${\rm Ru_{50}Nb_{50}}$ (see  \cite{vermautetal13,manzonietal14}).  In this connection we recall  the result in \cite[Theorem 4]{p35}, which implies that for cubic tetragonal transformations there is no homogeneous GYM $\nu=\sum_{j=1}^4\lambda_j\delta_{\bA_j}$ with $\bA_j\in K$ and $\bar\nu=(\eta_1^2\eta_2)^{\frac{1}{3}}\bI$. Hence for a grain $\om_i$ there is no GYM $(\nu_{\bx})_{\bx\in\om_i}$ with $\nu_{\bx}=\sum_{j=1}^4\lambda_j(\bx)\delta_{\bA_j}$ and corresponding underlying macroscopic deformation $\by$, supported on just four matrices $\bA_j\in K$ and satisfying $\by(\bx)=\ba+(\eta_1^2\eta_2)^{\frac{1}{3}}\bx$ for $\bx\in\partial\om_i$. Thus if a zero energy microstructure is supported on just four matrices, such as in a double laminate, then the grain boundary must deform nontrivially. 

The reduction to planar grain boundaries in Section \ref{3.1} assumes some regularity of the deformation or GYM as the point on the grain boundary is approached. Another possible approach is to use generalizations of the Hadamard jump condition as in \cite{u5} which do not require such regularity assumptions. Such an approach was carried out in \cite{p37}, but restricted to the special situation of a bicrystal  with cylindrical grains and two energy wells. Grain boundaries with a fractal-like microstructure on either side, with no obvious correlation between these microstructures, are observed in the alloy \Tn  (see \cite{u20}), for which the middle eigenvalue $\lambda_2(\bU)$ of the transformation strain is very close to 1. Similar complex microstructures at a grain boundary could be preferable for alloys satisfying the  cofactor conditions \cite{chenetal13}, which  include in particular the condition $\lambda_2(\bU)=1$.

Although the Taylor set in principle gives conditions under which a macroscopic deformation corresponds to a zero energy microstructure, the order of the laminates required to achieve this  (as, for example, in the self-accommodation analysis of Bhattacharya \cite{bhattacharya92}) will in general be unrealistically high. In particular the presence of such high order laminates could make the neglect of interfacial energy in the nonlinear elasticity model problematic (see e.g. R\"{u}land et al \cite{ruland25}).  In fact it does not seem to be known whether an arbitrary point $\bA\in K^{\rm qc}$ is the centre of mass $\bA=\bar\nu$ of a GYM $\nu=\sum_{j=1}^s\lambda_j\delta_{\bA_j}$ supported on a {\it finite number} of points of $K$, and if so whether there is a bound on the number $s$. Since the points of $\mS(K)$ are special points of $K^{\rm qc}$ it may be they can be represented in this form whereas general points of $K^{\rm qc}$ cannot.   In any case it would be interesting to develop an approximate model  in which  the microstructure in each grain is a low order laminate or supported on a specified maximum number of gradients. Such a model could give insight into correlations between the microstructures in adjacent grains and at longer length-scales.

 Finally it would be useful to extend aspects of our analysis to incorporate slip, which can also be represented by rank-one connections between energy wells according to the Ericksen energy well picture (see \cite{j73}), and applied tension (see e.g. \cite{karamietal25}).  For related experimental work see \cite{molnarova2021,Tyc2024}, and for a theoretical analysis  \cite{EnglKriesbeck}, each of these papers involving considerations related to the Taylor set.

\section{Appendix: Supplementary material for Section \ref{3.3}}\label{Appendix}
This MATLAB function computes symbolically the product of determinants \eqref{Fprod} in Section \ref{3.3}. Given a rotation angle $\varphi$, a skew-symmetric generator $\bK$ (so $\bR=\bI-\sin\varphi\,\bK+(1-\cos\varphi\,)\bK^2$ is the relative rotation), and the $3\times 3$ matrix $\bM$ defined by \eqref{Mm}, the function defines six standard rank‑one twinning systems according to Table \ref{Table} consisting of the vectors $\ba_j$, and normals $\bn_j,$ that can be on the one side of the boundary and fixes one of those solutions, namely $(\ba_1,\bn_1)$, on the other side of the boundary. For each of the six systems it constructs the symbolic $3\times 3$ block matrices $\bA_0,\bA_1,\bA_2,\bB_1,\bB_2$, given by \eqref{Melements}, assembles the  $4\times 4$ matrix \eqref{M} (here denoted as Mat and computes its determinant. The determinants are simplified using exact rational cancellation and polynomial degree checks in $\eta_2$ and multiplied into a global symbolic product. The outputs are DET (simplified symbolic product), DETvpa (numeric‑approximation display), and the last computed matrix Mat. The code is written for symbolic manipulation (uses sym, simplify, vpa) to expose algebraic factorisation and degree and structure of the resulting polynomials.

\begin{small}
\begin{Verbatim}[breaklines=true]
function [DET, DET_vpa, Mat] = Mdet(phi, K, M)
%Initialize variables \eta_1, \eta_2 (here n1, n2 for brevity)
syms n2 positive
n1=sym(1);
assumeAlso(n2 ~= n1)
%Transpose of relative Rotation R 
Rt = eye(3) + sin(phi)*K + (1-cos(phi))*(K^2);
%Define \rho
rho = (n2^2 - n1^2)/(n1^2 + n2^2);
%Define Tetragonal Wells 
U1 = diag([n2 n1 n1]);
U3 = diag([n1 n1 n2]);
% Placeholder  to switch between U1 and U3 
W1_well= {U1,U1,U3,U3,U3,U3};
%Array of values {aj, nj} that solve twinning equation
systems = { ...
[-sqrt(sym(2))*rho*n2; sqrt(sym(2))*rho*n1; 0], sym([1; 1; 0])/sqrt(sym(2)); ...
[-sqrt(sym(2))*rho*n2; -sqrt(sym(2))*rho*n1; 0], sym([1; -1; 0])/sqrt(sym(2)); ...
[-sqrt(sym(2))*rho*n2; 0; sqrt(sym(2))*rho*n1], sym([1; 0; 1])/sqrt(sym(2)); ...
[-sqrt(sym(2))*rho*n2; 0; -sqrt(sym(2))*rho*n1], sym([1; 0; -1])/sqrt(sym(2));  ...
[0; sqrt(sym(2))*rho*n1; -sqrt(sym(2))*rho*n2], sym([0; 1; 1])/sqrt(sym(2)); ...
[0; -sqrt(sym(2))*rho*n1; -sqrt(sym(2))*rho*n2], sym([0; -1; 1])/sqrt(sym(2))  ...
};
%We choose one of the twinning solutions on one side of the boundary
a1 = [-sqrt(sym(2))*rho*n2; sqrt(sym(2))*rho*n1; 0];
n11 = sym([1;1;0])/sqrt(sym(2));
%Initialise the product of determinants to be 1
GlobalDet = sym(1); 
for l = 1:6
ak = systems{l,1};
nk = systems{l,2};

%% COEFFICIENTS
A0 = U1^2 - Rt' * W1_well{l}^2 * Rt;
x = U1*a1;
A1 = x*n11' + n11*x';
A2 = (a1'*a1)*(n11*n11');
y = W1_well{l}*ak;
B1 = Rt'*(y*nk' + nk*y')*Rt;
B2 = (ak'*ak)*Rt'*(nk*nk')*Rt;
%% MATRIX
Mat = [
-trace(A1*M*B1*M*A0*M), -trace(A2*M*B1*M*A0*M), -trace(A1*M*B2*M*A0*M), -trace(A2*M*B2*M*A0*M);
-trace(A2*M*B1*M*A0*M), -trace(A2*M*B1*M*A1*M), -trace(A2*M*B2*M*A0*M), -trace(A2*M*B2*M*A1*M);
-trace(A1*M*B2*M*A0*M), -trace(A2*M*B2*M*A0*M),  trace(A1*M*B2*M*B1*M),  trace(A2*M*B2*M*B1*M);
-trace(A2*M*B2*M*A0*M), -trace(A2*M*B2*M*A1*M),  trace(A2*M*B2*M*B1*M),  0
];
Ri=det(Mat); 
Rri=simplify(Ri);
RiVPA=vpa(Rri, 3); 
pretty(RiVPA);
% Exact rational cancellation 
[N, D] = numden(Ri);        % numerator and denominator
g = gcd(N, D);              % symbolic gcd 
Ri_simpl = simplify(N./g) ./ simplify(D./g);
% degree in each variable separately
degN = polynomialDegree(N, n2);
degD = polynomialDegree(D, n2);
% Multiply the simplified symbolic Ri into the product 
GlobalDet = simplifyFraction(GlobalDet * Ri_simpl);
end
GlobalDet_simpl = simplifyFraction(GlobalDet);
DET=GlobalDet_simpl;
DET_vpa=vpa(GlobalDet_simpl, 3);
\end{Verbatim}
\end{small}

The MATLAB script that follows defines the inputs for the function above and computes the symbolic gcd (greatest common divisor) between the polynomials in \eqref{Poly2}. In particular, this script sets up two test cases and uses the function Mdet to compute symbolically the product of determinants \eqref{Fprod} for each case. It defines two prescribed rotation axes and angles (via skew-symmetric generators $\bK$ and $\bK_2$ and Rodrigues' formula), and two antisymmetric matrices $\bM$ and $\bM_2$ from specified axis vectors. For each case, it calls Mdet to build and simplify a product of $4\times 4$ determinant expressions (one for each of six standard rank‑one twin systems), producing a readable result as well as the full expression of the determinant. Using the later, we then compute the two exact polynomials $P_1$ and $P_2$ in $\eta_2$ given by \eqref{Poly2}. The code clears scalar denominators from two symbolic expressions  so that they become polynomials in the variable $\eta_2$ with no fractional coefficients. It first removes any global rational factor, then ensures every coefficient (as a function of $\eta_2$) is integer/rational by multiplying by a common multiple of their denominators. It further removes any common multiplicative factor yielding primitive polynomials, to finally compute the symbolic gcd  polynomial between  $P_1(\eta_2)$ and $P_2(\eta_2).$ It reports back the polynomial degree as well as whether the resulting gcd is constant or not.

\begin{small}
\begin{Verbatim}[breaklines=true]
syms n2 positive
n1=sym(1);
assumeAlso(n1 ~= n2)

%% ROTATION MATRICES
phi = sym(2*pi/3);
k = sym([1 2 7]);
k = k / norm(k);
k1 = k(1); k2 = k(2); k3 = k(3);
K = [ 0 -k3 k2;
k3 0 -k1;
-k2 k1 0 ];
psi = sym(pi/4);
z = sym([3 2 5]);
z = z / norm(z);
z1 = z(1); z2 = z(2); z3 = z(3);
K2 = [ 0 z3 -z2;
    -z3 0 z1;
    z2 -z1 0 ];
m=sym([1;2;3]/sqrt(14));
M = [ 0 -m(3) m(2);
    m(3) 0 -m(1);
    -m(2) m(1) 0 ];
m2=sym([7;5;1]/sqrt(75));
M2 = [ 0 -m2(3) m2(2);
    m2(3) 0 -m2(1);
    -m2(2) m2(1) 0 ];

%% DETERMINANTS
[Det1, Det1_vpa] = Mdet(phi, K, M);
[Det2, Det2_vpa] = Mdet(psi, K2, M2);
pretty(Det1_vpa);
pretty(Det2_vpa);

%% Exact symbolic polynomials
P1 = (n2^2+n1^2)^52 * Det1 / (n2^2-n1^2)^102;
P2 = (n2^2+n1^2)^52 * Det2 / (n2^2-n1^2)^102;
% ---- Remove scalar denominators ----
[numP1, denP1] = numden(P1);
[numP2, denP2] = numden(P2);
Lden = lcm(denP1, denP2); % common scalar denominator
P1 = expand(numP1 * (Lden / denP1));
P2 = expand(numP2 * (Lden / denP2));
c1 = coeffs(P1, n2, 'All');
c2 = coeffs(P2, n2, 'All');
D1 = sym(ones(1,numel(c1)));
D2 = sym(ones(1,numel(c2)));
for kk = 1:numel(c1), [~, D1(kk)] = numden(c1(kk)); end
for kk = 1:numel(c2), [~, D2(kk)] = numden(c2(kk)); end
Lcoeff = D1(1);
for kk = 2:numel(D1), Lcoeff = lcm(Lcoeff, D1(kk)); end
for kk = 1:numel(D2), Lcoeff = lcm(Lcoeff, D2(kk)); end
P1 = expand(P1 * (Lcoeff));
P2 = expand(P2 * (Lcoeff));

% ---- Divide out common scalar content of coefficients ----
allCoeffs = [coeffs(P1, n2, 'All'), coeffs(P2, n2, 'All')];
scalarG = allCoeffs(1);
for kk = 2:numel(allCoeffs)
    scalarG = gcd(scalarG, allCoeffs(kk));
end
if ~isequal(scalarG, sym(1))
    P1 = simplify(P1 / scalarG);
    P2 = simplify(P2 / scalarG);
    fprintf('Divided out scalar content from P1/P2.\n');
end

%% Exact symbolic GCD in n2 (covers ± roots)
Gsym = gcd(P1, P2, n2);
degG = double(polynomialDegree(Gsym, n2));
degP1 = double(polynomialDegree(P1, n2));
degP2 = double(polynomialDegree(P2, n2));

% ---- Report polynomial degrees, GCD & GCD degree ----
fprintf('Degree:  P1 = %d, P2 = %d\n', degP1, degP2);
fprintf('GCD degree (in n2) = %d\n', degG);
if isequal(Gsym, sym(1)) || (degG == 0 && ~isequal(Gsym, sym(0)))
    fprintf('Exact symbolic gcd is constant; thus no polynomial common root (in n2).GCD is %s\n', char(Gsym));
else
    fprintf('Nontrivial polynomial gcd found (degree %d).\n', degG);
    % display compact gcd
    fprintf('Gcd is: %s\n', char(Gsym));
end
\end{Verbatim}
\end{small}

\section*{AI and computational resources disclosure} The authors acknowledge a very useful interaction with Gemini 3 concerning the proof of Theorem \ref{compatthm}, specifically the use of the Dixon resultant and the formulation of Lemma \ref{LemmaM}. The symbolic computations were carried out using MATLAB.

\section*{Ethics declaration}  Not applicable.

\section*{Acknowledgments}
We are grateful to Kaushik Bhattacharya, Michael Peigney, Damian Rosler and Hanu\v{s} Seiner for their interest and useful discussions, and to the referees for helpful suggestions that improved the paper. The article was mainly written while MG was a Postdoctoral Research Associate in the Department of Mathematics at Heriot-Watt University under the   project {\it Mathematical theory of polycrystalline materials} supported by the EPSRC grant EP/V00204X. JMB was supported by the same EPSRC grant,  by a grant from the Institute
for Advanced Study School of Mathematics, and through his appointment as a Senior Fellow at the Hong Kong Institute for Advanced Study, City University of Hong Kong.
	
 \bibliographystyle{plain}
 
\bibliography{balljourn,ballconfproc,ballmisc,gen3,ballprep}
	 
\end{document}